\documentclass[11pt]{article}
\usepackage[all,pdf]{xy}

\usepackage{amsmath,amssymb,amsthm}
\usepackage{hyperref}
\usepackage{mathrsfs}
\usepackage{graphicx}

\usepackage{tikz}
\usepackage{titlesec}
\usepackage{float}

\usepackage{amssymb}

\theoremstyle{plain}
\newtheorem{theorem}{Theorem}[section]
\newtheorem{lemma}[theorem]{Lemma}

\newtheorem{proposition}[theorem]{Proposition}
\newtheorem{corollary}[theorem]{Corollary}

\theoremstyle{definition}
\newtheorem{definition}[theorem]{Definition}
\newtheorem{remark}[theorem]{Remark}
\newtheorem{example}[theorem]{Example}

\usetikzlibrary{shapes.geometric,arrows}

\newcommand{\ua}{\mathord{\uparrow}}
\newcommand{\da}{\mathord{\downarrow}}
\newcommand{\rom}[1]{\rm{\uppercase\expandafter{\romannumeral #1}}}

\newcommand{\cl}{\mathrm{cl}}

\makeatletter
\def\ps@pprintTitle{%
  \let\@oddhead\@empty
  \let\@evenhead\@empty
  \def\@oddfoot{\reset@font\hfil\thepage\hfil}
  \let\@evenfoot\@oddfoot
}
\makeatother

\title{
The category of well-filtered dcpo's is not $\Gamma$-faithful
}

\author{Hualin Miao, Huijun Hou, Xiaodong Jia\footnote{Corresponding author, Email: jiaxiaodong@hnu.edu.cn} ~and Qingguo Li \\
\emph{School of Mathematics, Hunan University, Changsha, Hunan 410082, China}}

\begin{document}
\maketitle

\begin{abstract}

The Ho-Zhao problem asks whether any two dcpo's with isomorphic Scott closed set lattices are themselves isomorphic, that is, whether the category $\mathbf{DCPO}$ of dcpo's and Scott-continuous maps is \emph{$\Gamma$-faithful}.  In 2018, 
Ho, Goubault-Larrecq, Jung and Xi answered this question in the negative, and they introduced the category $\mathbf{DOMI}$ of dominated dcpo's and proved that it is {$\Gamma$-faithful}. Dominated dcpo's subsume many familiar families of dcpo's in domain theory, such as the category of bounded-complete dcpo's and that of sober dcpo's, among others. However, it is unknown whether the category of dominated dcpo's dominates all well-filtered dcpo's, a class strictly larger than that of bounded-complete lattices and that of sober dcpo's. In this paper, we address this very natural question and show that the category $\mathbf{WF}$ of well-filtered dcpo's is not $\Gamma$-faithful, and as a result of it, well-filtered dcpo's need not be dominated in general. Since not all dcpo's are well-filtered, our work refines the results of Ho, Goubault-Larrecq, Jung and Xi. 

As a second contribution, we confirm that the  Lawson's category of $\Omega^{*}$-compact dcpo's is $\Gamma$-faithful. Moreover, we locate a class of dcpo's which we call weakly dominated dcpo's, and show that this class is  $\Gamma$-faithful and strictly larger than $\mathbf{DOMI}$.

\emph{Keywords}: Ho-Zhao Problem, $\Gamma$-faithfulness, well-filtered dcpo's, weakly dominated dcpo's.

\end{abstract}

\section{Introduction}
One of the most important topologies on posets is the so-called Scott topology, which consists of upper subsets that are inaccessible by suprema of directed subsets in given posets, and has been playing prominent r\^ole in domain theory, non-Hausdorff topology and denotational semantics, among others. As can be seen from definition, Scott topologies are uniquely determined by the order structure on posets. A more intriguing question is the converse: does the Scott topology of a poset determine the order of the poset? In order to form this question more rigorously, we let $\Gamma P$ denote the set of all Scott closed subsets of a poset $P$, ordered by inclusion. The poset $\Gamma P$ is a complete lattice for each $P$. The aforementioned question is then formalized as:
\begin{center}
\emph{If $P$ and $Q$ are posets with $\Gamma P$ isomorphic to $\Gamma Q$ ($\Gamma P \cong \Gamma Q$, in symbols), is it true that $P\cong Q$?}
\end{center}
This question was perfectly answered by Zhao and Fan in \cite{dp}, where they showed that each poset~$P$ admits a so-called dcpo-completion $E(P)$, with the property that $E(P)$ is a directly-complete poset (dcpo for short) and $\Gamma P \cong \Gamma E(P)$. Hence, any poset which fails to be a dcpo, together with its dcpo-completion~$E(P)$, fails the above question.  This observation has led Ho and Zhao narrow this question to the category $\mathbf{DCPO}$ of dcpo's and Scott-continuous maps~\cite{lost}, and they call a full subcategory~$\mathbf C$ of $\mathbf{DCPO}$ $\Gamma$-faithful if for every pair $P$ and $Q$ in $\mathbf C$, $\Gamma P \cong \Gamma Q$ implies $P\cong Q$. In \cite{lost}, Ho and Zhao identified the category of bounded complete dcpo's as one of $\Gamma$-faithful subcategories of $\mathbf{DCPO}$, and asked whether the category $\mathbf{DCPO}$ itself is $\Gamma$-faithful. The question was later dubbed the \emph{Ho-Zhao} problem. In 2016, Ho, Goubault-Larrecq, Jung and Xi exhibited a dcpo~$\mathcal H$ which is not sober in the Scott topology, and showed that $\mathcal H$ and its sobrification $\hat{\mathcal H} $ serve as instances for the fact that $\mathbf{DCPO}$ fails to be $\Gamma$-faithful. In addition, they gave a $\Gamma$-faithful full subcategory $\mathbf{DOMI}$ of $\mathbf{DCPO}$, which consists of the so-called \emph{dominated dcpo's}. Notably, the category $\mathbf{DOMI}$ of dominated dcpo's subsumes many known $\Gamma$-faithful subcategories of dcpo's, for example, the category of sober dcpo's and that of bounded-complete lattices, to name a few. In Section $3$, we show that the category $\mathbf{DOMI}$ of dominated dcpo's also includes the category of $\Omega^{*}$-compact dcpo's. 

Generalizing both sober dcpo's and $\Omega^{*}$-compact dcpo's, \emph{well-filtered} dcpo's are ones that are well-filtered in the Scott topology, and the category $\mathbf{WF}$ of well-filtered dcpo's is a strictly smaller subclass of $\mathbf{DCPO}$. A natural question arises as whether well-filtered dcpo's also are subsumed under dominated dcpo's, or whether the category $\mathbf{WF}$ is $\Gamma$-faithful?
In this paper, we will mainly investigate this natural question and give concrete examples to deduce that $\mathbf{WF}$ is not $\Gamma$-faithful. Hence, well-filtered dcpo's may fail to be dominated. Our examples make use of an example $Y$ given by Zhao and Xi in \cite{wns}, where they showed that $Y$ is a well-filtered dcpo, but not sober in its Scott topology. 

While the paper itself displays a negative result, it refines the results of Ho, Goubault-Larrecq, Jung and Xi, and it also reveals the distinction between the category $\mathbf{WF}$ and the newly found $\mathbf{DOMI}$, where the later class $\mathbf{DOMI}$ definitely needs more clarification.  Finally, we introduce the category of weak dominated dcpo's, and show that it is a $\Gamma$-faithful category that is strictly larger than the class of dominated dcpo's.

\section{Preliminaries}

In this section, we introduce some basic concepts and notations that will be used in this paper.

Let $P$ be a partially ordered set (\emph{poset}, for short), $D\subseteq P$ is \emph{directed} (resp., \emph{filtered}) if $D$ is nonempty and for any finite subset $F\subseteq D$, there is a $d\in D$ such that $d$ is an upper bound (resp., a lower bound) of $F$. A poset $P$ is called \emph{directed complete} (\emph{dcpo}, for short) if every directed subset $D$ of $P$ has a least upper bound, which we denote by $\sup D$, or $\bigvee D$. For any subset $A\subseteq P$, let $\uparrow$$A$ = $\{x\in P: x\geq$ a for some $a\in A\}$ and $\downarrow$$A$ = $\{x\in P: x\leq a$ for some $a\in A\}$. Specifically, we write $\uparrow$$x$ = $\uparrow$$\{x\}$ and $\downarrow$$x$ = $\downarrow$$\{x\}$. We will call $A\subseteq P$ an \emph{upper set} (resp., a \emph{lower set}) if $A$ = $\uparrow$$A$ (resp., $A$ = $\downarrow$$A$).

A subset $U$ of $P$ is \emph{Scott open} if $U$ = $\uparrow$$U$ and for any directed subset $D$ for which sup$D$ exists, sup$D$ $\in U$ implies $D\cap U \neq \emptyset$. Accordingly, $A\subseteq P$ is \emph{Scott closed} if $A$ = $\downarrow$$A$ and for any directed subset $D$ of $P$ with sup$D$ existing, $D\subseteq A$ implies sup$D$ $\in A$. The set of all Scott open sets of $P$ forms the \emph{Scott topology} on $P$, which is denoted by $\sigma P$, and the set of all Scott closed sets of $P$ is denoted by $\Gamma P$. Furthermore, for a subset $A$ of $P$, we will use $\overline{A}$ or $\cl(A)$ to denote the closure of $A$ with respect to the Scott topology on $P$. The space $(P, \sigma P)$ also is denoted by $\Sigma P$, some authors refer such a space, a poset endowed with the Scott topology, a \emph{Scott space}.
Recall that the lower topology $\omega (L)$ on $L$ is defined to have the principal filters $\ua a$ for $a\in L$ as subbasic closed sets.

For a $T_{0}$ space $X$, the partial order $\leq$$_{X}$, defined by $x\leq_{X} y$ if and only if $x$ is in the closure of $y$, is called the \emph{specialization order} on $X$. Naturally, the closure of a single point $x$ is $\downarrow$$x$, the order considered here is, of course, the specialization order. The specialization order on a Scott space $\Sigma P$ coincides with the original order on $P$. 
A subset $K$ of $X$ is \emph{compact} if every open cover of $K$ admits a finite subcover, and the \emph{saturation} of a subset $A$ is the intersection of all open sets that contain it. A subset $K$ is called to be \emph{saturated} if it equals its saturation. Besides, for every subset $A$ of $X$, the saturation of $A$ coincides with $\uparrow$$A$ under the specialization order. 
Every compact saturated subset $K$ in a Scott space is of the form $\ua \mathrm{min}K$, where $\mathrm{min}K$ is the set of minimal elements of $K$~\cite{jia}. 
And usually, we use $Q(X)$ to denote the set of all non-empty compact saturated subsets of $X$. 

A $T_0$ space $X$ is \emph{sober} if every irreducible closed subset $C$ of $X$ is the closure of some unique singleton set $\{c\}$, where $C$ is called \emph{irreducible} if $C\subseteq A\cup B$ for closed subsets $A$ and $B$ implies that $C\subseteq A$ or $C\subseteq B$. We denote the set of all closed irreducible subsets of $X$ by $IRR(X)$. 
Every sober space $X$ is a \emph{well-filtered} space, in the sense that for every filtered family (in the conclusion order) of compact saturated subsets $K_i, i\in I$ and every open subset $U$ of $X$,  $\bigcap_{i\in I}K_i\subseteq U$ implies that there is already some $K_i$ with $K_i\subseteq U$.
Well-filtered spaces can also be characterized by the so-called \emph{KF}-sets. In a $T_{0}$ space~$X$, a nonempty subset $A$ of $X$ is said to have the \emph{compactly filtered property} ($KF$ property), if there exists a filtered family $\mathcal{K}\subseteq_{\mathrm{flt}} Q(X)$ such that $\mathrm{cl}(A)$ is a minimal closed set that intersects all members of $\mathcal{K}$. We call such a set a $KF$-set and denote by $KF(X)$ the set of all closed $KF$-sets of $X$. It is shown in~\cite{owr} that a $T_0$ space is well-filtered if and only if every $KF$-set $K$ of $X$ is the closure of some unique singleton set $\{k\}$. Finally, we remark that every $KF$-set is irreducible, i.e., $KF(X) \subseteq IRR(X)$ for each $T_0$ space $X$~\cite{owr}.

\section{$\Omega^{*}$-compact dcpo's are dominated}
In this section, we deduce that the category $\mathbf{DOMI}$ of dominated dcpo's subsumes the category of $\Omega^{*}$-compact dcpo's. 
\begin{definition}\cite{LWX}
	Let $X$ be a $T_{0}$ space and $\leq$  be its order of specialization. We equip X also with its $\omega$-topology defined from the order of specialization. We say that $X$ is \emph{$\Omega^{*}$-compact} if every closed subset of $X$ is compact in the $\omega$-topology.
\end{definition}
\begin{definition}\cite{hzp}
Given $A,B\in IRR(L)$, we write $A\lhd B$ if there is $b\in  B$ such that $A \subseteq\da b$. We
write $\nabla B$ for the set $\{A\in IRR(L)\mid A\lhd B\}$. A dcpo $L$ is called dominated if for every closed irreducible subset $A$ of $L$, the collection $\nabla A$ is Scott closed in $IRR(L)$.
\end{definition}
\begin{lemma}
	Let $L$ be an $\Omega^{*}$-compact dcpo. Then $L$ is a dominated dcpo.
	\begin{proof}
	From the definition of dominated dcpo's,	it suffices to prove that $\nabla A$ is a Scott closed subset of $IRR(L)$ for any $A\in IRR(L)$. To this end, let $(A_{i})_{i\in I}$ be a directed subset of $\nabla A$. Then there exists $a_{i}\in A$ such that $A_{i}\subseteq \da a_{i}$ for any $i\in I$. This means that $(\bigcap_{x\in A_{i}}\ua x)\cap A\neq \emptyset$ for any $i\in I$. Note that $\bigcap_{x\in A_{i}}\ua x$ is closed in the lower topology and $(\bigcap_{x\in A_{i}}\ua x)_{i\in I}$ is filtered. Then the fact that $L$ is $\Omega^{*}$-compact reveals that $A\cap \bigcap_{x\in \bigcup_{i\in I}A_{i}}\ua x\neq \emptyset$. Since $\bigcap_{x\in cl(\bigcup_{i\in I}A_{i})}\ua x=\bigcap_{x\in \bigcup_{i\in I}A_{i}}\ua x$, we have $A\cap (\bigcap_{x\in cl(\bigcup_{i\in I}A_{i})}\ua x)\neq \emptyset$. Hence, $\sup_{i\in I}A_{i} = cl(\bigcup_{i\in I}A_{i})\in \nabla L$. 
	\end{proof}
\end{lemma}
From the above lemma, we can arrive at the following theorem immediately.
\begin{theorem}
The category of $\Omega^*$-compact dcpo's is $\Gamma$-faithful.
\end{theorem}
We know that $\Omega^*$-compact dcpo's are well-filtered dcpo's. It is nature to ask whether the category $\mathbf{WF}$ of well-filtered dcpo's is $\Gamma$-faithful. Next, we study the problem.
\section{The category of well-filtered dcpo's is not $\Gamma$-faithful}
In this section, we first give the definition of a dcpo~$\mathcal Z$, which will be shown to be well-filtered but not sober in the Scott topology. Moreover, we will see that $\Gamma \mathcal Z$ is isomorphic to $\Gamma \hat{\mathcal Z}$, where $\hat{\mathcal Z}$ is the poset of all irreducible closed subsets of $\mathcal Z$ in the set inclusion order, and $\hat{\mathcal Z}$ is a sober dcpo. Since $\mathcal Z$ is not sober, it cannot be isomorphic to $\hat{\mathcal Z}$. Hence, the pair of well-filtered dcpo's $\mathcal Z$ and $\hat{\mathcal Z}$ illustrates that the category of well-filtered dcpo's is not $\Gamma$-faithful.
\subsection{The definition of the counterexample $\mathcal{Z}$}
Let $\omega_{1}$ be the first non-countable ordinal and $\mathbb{W}= [0, \omega_{1})$ be the set of all ordinals strictly less than $\omega_{1}$. Then $\mathbb{W}$ consists of all finite and infinite countable ordinals.

\begin{remark}\cite{gt}
The following results about $\mathbb{W}$ are well-known:
\begin{enumerate}
\item $|\mathbb{W}| = \aleph_{1}$, where $|\mathbb{W}|$ denotes the cardinality of $\mathbb{W}$.

\item $\mathbb W$ is sequentially complete. That is, for every countable subsequence $D\subseteq \mathbb{W}$, $\sup D\in \mathbb{W}$, here the $\sup D$ is taken with respect to the usual linear order on ordinals.

\item For any $\alpha\in \mathbb{W}$, $\{\beta: \beta\leq\alpha\}$ is a finite or countably infinite subset of $\mathbb W$.
\end{enumerate}
\end{remark}

Now let us review a poset $Y$ given by Zhao and Xi in~\cite{wns}. We will use this poset as building blocks for our dcpo $\mathcal Z$.  Let $\mathbb{W}_{1} = [0, \omega_{1}] = \mathbb{W}\cup \{\omega_{1}\}$ and $Y = \mathbb{W}\times \mathbb{W}_{1}$. The order on $Y$ is given by $(m, u) \leq (m', u')$ if and only if: 

\begin{itemize}
\item $m = m'$ and $u\leq u'$, or
\item $u' = \omega_{1}$ and  if $u\leq m'$,
\end{itemize}
and the order structure of $Y$ can be easily depicted, as in Figure $1$.

\begin{figure}[H]
\centering
\begin{tikzpicture}[line width=0.7pt,scale=1]
\fill[black]  (0,0) circle (1.8pt);
\fill[black] (0,1) circle (1.8pt);
\fill[black] (0,2.3) circle (1.8pt);
\fill[black] (0,3.3) circle (1.8pt);
\fill[black] (0,5) circle (1.8pt);
\draw (0,0)--(0,1);
\draw (0,2.3)--(0,3.3);
\draw [dashed](0,3.3)--(0,5);
\draw [dashed](0,1)--(0,2.3);

\fill[black] (1.7,0) circle (1.8pt);
\fill[black] (1.7,1) circle (1.8pt);
\fill[black] (1.7,2.3) circle (1.8pt);
\fill[black] (1.7,3.3) circle (1.8pt);
\fill[black] (1.7,5) circle (1.8pt);
\draw (1.7,0)--(1.7,1);
\draw (1.7,2.3)--(1.7,3.3);
\draw [dashed](1.7,3.3)--(1.7,5);
\draw [dashed](1.7,1)--(1.7,2.3);

\fill[black] (5,0) circle (1.8pt);
\fill[black] (5,1) circle (1.8pt);
\fill[black] (5,2.3) circle (1.8pt);
\fill[black] (5,3.3) circle (1.8pt);
\fill[black] (5,5) circle (1.8pt);
\draw (5,0)--(5,1);
\draw (5,2.3)--(5,3.3);
\draw [dashed](5,3.3)--(5,5);
\draw [dashed](5,1)--(5,2.3);

\fill[black] (7.4,0) circle (1.8pt);
\fill[black] (7.4,1) circle (1.8pt);
\fill[black] (7.4,2.3) circle (1.8pt);
\fill[black] (7.4,3.3) circle (1.8pt);
\fill[black] (7.4,5) circle (1.8pt);
\draw (7.4,0)--(7.4,1);
\draw (7.4,2.3)--(7.4,3.3);
\draw [dashed](7.4,3.3)--(7.4,5);
\draw [dashed](7.4,1)--(7.4,2.3);

\draw [dashed](2.1,0)--(3.3,0);
\draw [dashed](2.1,1)--(3.3,1);
\draw [dashed](2.1,2.3)--(3.3,2.3);
\draw [dashed](2.1,3.3)--(3.3,3.3);
\draw [dashed](2.1,5)--(3.3,5);

\draw [dashed](8.3,-0.3)--(9.4,-0.3);
\draw [dashed](8.3,0.3)--(9.4,0.3);
\draw [dashed](8.3,1.3)--(9.4,1.3);
\draw [dashed](8.3,0.7)--(9.4,0.7);
\draw [dashed](8.3,2)--(9.4,2);
\draw [dashed](8.3,2.6)--(9.4,2.6);
\draw [dashed](8.3,3.6)--(9.4,3.6);
\draw [dashed](8.3,3)--(9.4,3);
\draw [dashed](8.3,5)--(9.4,5);

\draw (-1,0.3)--(8,0.3)   (-1,-0.3)--(8,-0.3);
\draw (-1,0.7)--(8,0.7)   (-1,1.3)--(8,1.3);
\draw (-1,2)--(8,2)   (-1,2.6)--(8,2.6);
\draw (-1,3)--(8,3)   (-1,3.6)--(8,3.6);

\draw (-1,0.3) arc (90:270:0.3);
\draw (-1,1.3) arc (90:270:0.3);
\draw (-1,2.6) arc (90:270:0.3);
\draw (-1,3.6) arc (90:270:0.3);

\draw (0,5)--(0.8,0.3)  (1.7,5)--(2.4,1.3)  (5,5)--(5.6,2.6)  (7.4,5)--(7.8,3.6);

\node[left][font=\tiny]  at(0,0) {(0,0)};
\node[left][font=\tiny]  at(0,1) {(0,1)};
\node[left][font=\tiny]  at(0,2.3) {(0,$\omega_{0}\!$)};
\node[left][font=\tiny]  at(0,3.3) {($\!0$,$\omega_{0}\!\!+\!\!1\!$)};
\node[left][font=\tiny]  at(0,5) {(0,$\omega_{1}\!$)};

\node[left][font=\tiny]  at(1.7,0) {(1,0)};
\node[left][font=\tiny]  at(1.7,1) {(1,1)};
\node[left][font=\tiny]  at(1.7,2.3) {(1,$\omega_{0}\!$)};
\node[left][font=\tiny]  at(1.7,3.3) {$(\!1\!,\!\omega_{0}\!\!+\!\!1\!)$};
\node[left][font=\tiny]  at(1.7,5) {(1,$\omega_{1}\!$)};

\node[left][font=\tiny]  at(5,0) {($\omega_{0}$,0)};
\node[left][font=\tiny]  at(5,1) {($\omega_{0}$,1)};
\node[left][font=\tiny]  at(5,2.3) {($\omega_{0}$,$\omega_{0}\!$)};
\node[left][font=\tiny]  at(5,3.3) {$(\!\omega_{0}\!,\!\omega_{0}\!\!+\!\!1\!)$};
\node[left][font=\tiny]  at(5,5) {($\omega_{0}$,$\omega_{1}\!$)};

\node[left][font=\tiny]  at(7.4,0) {($\!\omega_{0}\!\!+\!\!1\!$,0)};
\node[left][font=\tiny]  at(7.4,1) {($\!\omega_{0}\!\!+\!\!1\!$,1)};
\node[left][font=\tiny]  at(7.4,2.3) {($\!\omega_{0}\!\!+\!\!1\!$,$\omega_{0}\!$)};
\node[left][font=\tiny]  at(7.4,3.3) {$(\!\omega_{0}\!\!+\!\!1\!,\!\omega_{0}\!\!+\!\!1\!)$};
\node[left][font=\tiny]  at(7.4,5) {($\!\omega_{0}\!\!+\!\!1\!$,$\omega_{1}\!$)};

\end{tikzpicture}
\scriptsize \\ Figure 1: A non-sober well-filtered dcpo $Y$. 
\end{figure}

We gather some known results about $Y$. The reader can find details in~\cite{wns}.
\begin{lemma}
\begin{enumerate}
\item The poset $Y$ is a dcpo;
\item $Y$ is well-filtered, but not sober in the Scott topology;
\item The only closed irreducible subset of $Y$ that is not a principal ideal is $Y$ itself.
\end{enumerate}
\end{lemma}

For each $u \in [0, \omega_1)$, we use the convention $L_u$ to denote the set of elements on the \emph{$u$-th level} of $Y$. That is, $L_u = \{(x, u)\in Y \mid x\in [0, \omega_1) \}$. We call $L_{\omega_1}$ the \emph{maximal level} in $Y$. Similarly, we say that elements of the form $(u, x), x\in [0, \omega_1]$, are in the \emph{$u$-th column} of $Y$. 

The process of constructing $\mathcal{Z}$ is similar to that of $\mathcal{H}$ in \cite{hzp}.  An informal description of this construction is as follows: 
We begin with $Y$ and for each level $L_{u}$ of $Y$, we add a copy $Y_{u}$ of $Y$ below $L_u$ and identify maximal points of $Y_u$ with elements in $L_u$ in the canonical way -- identifying $(x, \omega_1)$ in $Y_u$ with $(x, u)$ in $L_u$.
No order relation between the non-maximal elements of two different $Y_u, Y_{u'}$ is introduced. Now, we repeat this copy-and-identify process and add copies of $Y$ below non-maximal levels of $Y_u, u\in [0, \omega_1)$, copies of $Y$ below non-maximal levels of these copies of $Y$ which are already added in the previous step, and proceed infinitely countable many times.

 In order to keep all copies of $Y$ in right place, we use strings on $\mathbb W=[0, \omega_1)$ to index them. To start with, we let $\mathbb{W}^{*}$ be the set of finite strings of elements in $[0, \omega_1)$, and for any $s, t\in \mathbb{W}^{*}, u\in \mathbb{W}$, $u.s$ is such that adding the element $u$ to the front of $s$, and $ts$ is the concatenation of $t$ and $s$.
 Now, we use $Y_\varepsilon$ to denote the original $Y$, where $\varepsilon$ is the empty string; for strings $s$ of length~$1$, i.e., $s\in [0, \omega_1)$, $Y_s$ denotes the copy of $Y$ that is attached below $L_s$ of $Y_\varepsilon$, the $s$-th level of $Y_\varepsilon$; for strings $s$ of length larger than or equal to~$2$, say $s = x_1.s'$ with $x_1\in [0, \omega_1)$ and $s'$ being the obvious  remaining substring of $s$, $Y_s$ denotes the copy of $Y$ that is attached below the $x_1$-level of $Y_{s'}$ which is already well placed by induction. Then, the union $\bigcup_{s\in \mathbb{W}^*} Y_s$ will be the underlying set of our poset $\mathcal Z$. However, there are two issues to be settled. First, we need to denote elements in this big union. This is easy, as we can use $(m, u, s)$ with $(m, u, s) \in  [0, \omega_1)\times  [0, \omega_1] \times \mathbb W^*$ to denote the element $(m, u)$ in $Y_s$ in the canonical way. Second, we need to take the identifying process into the consideration. In this scenario, we will identify $(m, u, s), u \neq \omega_1, s \neq \varepsilon$, the element $(m, u)$ in $Y_s$, with $(m, \omega_1, u.s)$, the maximal point $(m, \omega_1)$ in $Y_{u.s}$. What we get after this identification is the correct underlying set of our poset $\mathcal Z$. 

Note that we equate the elements $(m, u, s)$ of $Y_{s}$ with $(m, \omega_{1}, u.s)$ of $Y_{u.s}$, each element of $\mathcal{Z}$ can be marked as the form of $(m, \omega_{1}, s)$, $m\in \mathbb{W}, s\in \mathbb{W}^{*}$, and $s$ may be $\varepsilon$. As $\omega_{1}$ appears as the second coordinate of each element in $\mathcal{Z}$, we simply omit it. By doing so it enables us to use $\{(m, s): (m, s)\in \mathbb{W}\times \mathbb{W}^{*}\}$ to label all elements of $\mathcal{Z}$.
With this all in mind, we now give the rigorous definition of $\mathcal Z$.

Let $\mathcal Z$ be $ \mathbb{W}\times \mathbb{W}^{*}$. We define the following relations on $\mathcal Z$ (where $m, m', u, u' \in \mathbb W, s, t \in \mathbb W^*$), which is reminiscent of the order on $\mathcal H$ in \cite{hzp}.

\begin{itemize}
\item $(m, u.s) <_{1} (m, u'.s)$ if $u < u'$;
\item $(m, ts) <_{2} (m, s)$ if $t\neq \varepsilon$;
\item $(m, ts) <_{3} (m', s)$ if $t\neq \varepsilon$ and $\mathrm{min}(t)\leq m'$.
\end{itemize}

In the third dotted item, $\mathrm{min}(t)$ denotes the least ordinal appearing in the string~$t$, hence $\mathrm{min}(t)\leq m'$ makes sense.
The definitions of these three kinds of relations are similar to that on $\mathcal{H}$ given by Ho et al. in \cite{hzp}, and we can obtain a partial order through the following results, where we use $``;"$ to mean composition of relations. 

\begin{proposition}\label{order}
\
\begin{enumerate}
\item $<_{1}, <_{2}, <_{3}$ are transitive and irreflexive, respectively.
\item $<_{1}; <_{2}\; =\; <_{2}$
\item $<_{1}; <_{3}\;\subseteq\; <_{3}$
\item $<_{2}; <_{3}\;\subseteq\; <_{3}$
\item $<_{3}; <_{2} \;\subseteq\; <_{3}$
\item $<\;:= \;<_{1}\cup <_{2}\cup <_{3}\cup\;(<_{2}; <_{1})\cup (<_{3};<_{1})$ is transitive and irreflexive.
\item $\leq \;:= (<\cup =)$ is a partial order relation on $\mathcal Z$.
\end{enumerate}
\end{proposition}

\begin{definition}
We define $\mathcal{Z}$ to be the poset which consists of the set $\mathbb{W}\times \mathbb{W}^{*}$ and the order relation $\leq$ in~\ref{order}. 
\end{definition}

In the rest of this section, we will prove in steps that $\mathcal{Z}$ is a well-filtered dcpo and that the pair $\mathcal{Z}$ and its sobrification $\hat{\mathcal{Z}}$ witness the fact that the category of well-filtered $\mathrm{dcpo's}$ is not $\Gamma$-faithful.

\subsection{$\mathcal{Z}$ is a dcpo}

\emph{Sketch of the proof.} We prove that $\mathcal{Z}$ is a dcpo by showing that every chain $C$ in $\mathcal Z$ has a supremum~\cite[Corollary 2]{markowsky76}.
Since $\sup C$ always exists when $C$ itself has a largest element, we only need to take care of the 
case where $C$ is a strictly increasing chain without a largest element. And of course, in the latter case $C$ must be infinite and we call such chains
\emph{non-trivial}. Similar to Proposition~5.3 in \cite{hzp}, we will be able to show that every non-trivial chain in $\mathcal{Z}$ contains a cofinal chain of the form $(m, u.s)_{u\in W}$, where $m\in \mathbb{W}$ and $s\in \mathbb{W}^{*}$ are fixed, and $W$ is an infinite subset of $\mathbb{W}$. This means that every non-trivial chain in $\mathcal{Z}$ eventually stays in the $m$-column of copy $Y_s$ for some string $s$. 
Hence, in order to prove that $\mathcal Z$ is a dcpo, we only need to verify the existence of suprema of non-trivial chains in that particular form, which is then doable. 

%


\begin{proposition}\label{dcpo}
$\mathcal{Z}$ is a $dcpo$.
\begin{proof}
Let $C$ be a non-trivial chain in $\mathcal{Z}$.

$\mathbf{Claim ~1}$: $C$ contains a cofinal chain of the form $(m, u.s)_{u\in W}$, with $m$ and $s$ fixed. 

Let $(m_1, s_1) < (m_2, s_2) < \cdots$ be a non-trivial chain. Each relationship $(m_i, s_i) <(m_{i+1}, s_{i+1})$ has to be one of the five types listed in item (6) of Proposition \ref{order}. All of these, except $<_1$, strictly reduce the length of the string $s_i$, so they can occur only finitely often along the chain. Therefore, from some index $i_0$ onward, the connecting relationship must always be $<_1$ which implies that the shape of the entries $(m_i, s_i)$, $i \geq i_0$, is as stated.

$\mathbf{Claim ~2}$: If $W$ is uncountable, then $\sup_{u\in W} (m, u.s) = (m, s)$.

Since $(m, u.s)<_{2} (m, s)$ for any $u\in W$, $(m, s)$ is an upper bound of $\{(m,u.s): u\in W\}$. It remains to confirm that $(m,s)$ is the least upper bound of $\{(m,u.s): u\in W\}$.
Assume that $(m', s')$ is another upper bound of $\{(m,u.s): u\in W\}$. Then each $(m, u.s)$ must be related to it by one of the five types listed in Proposition \ref{order}. We set
\begin{center}
$A_{<_{r}} = \{(m, u.s): u\in W\ \mathrm{and}\ (m, u.s)<_{r} (m', s')\}$,
\end{center}
where $<_{r}\in \{<_{1},\ <_{2},\ <_{3},\ <_{2};<_{1},\ <_{3};<_{1}\}$. It follows that at least one of the five sets is uncountable from the uncountability of $W$.

Case 1, $A_{<_{1}}$ is uncountable. For any $(m, u.s)\in A_{<_{1}}$, $(m, u.s)<_{1} (m', s')$, then $m = m'$ and $s'$ can be written as $u'.s$ for some fixed $u'\in \mathbb{W}$. It is obvious that $\downarrow\! u'$ is countable, which contradicts the fact that $A_{<_{1}}$ is uncountable.

Case 2, $A_{<_{2}}$ is uncountable. For any $(m, u.s)\in A_{<_{2}}$, $(m, u.s)<_{2} (m', s')$, so we have $m = m'$ and $s'\subseteq s$. Assume that $s = ts'$, where $t\in \mathbb{W^{*}}$. Then $(m, s) = (m', s')$ when $t = \varepsilon$ or $(m, s) <_{2} (m', s')$ when $t\neq\varepsilon$.

Case 3, $A_{<_{2};<_{1}}$ is uncountable. For any $(m, u.s)\in A_{<_{2};<_{1}}$, $(m, u.s)<_{2};<_{1} (m', s')$, that is, there exists $(m_{u}, s_{u})\in \mathcal{Z}$ such that $(m, u.s)<_{2} (m_{u}, s_{u})<_{1} (m', s')$. Then $m = m_{u} = m'$ and $s_{u}$ can be fixed since all of these $s_{u}$ have the same lengths as $s'$, and $s_{u}\subseteq s$. Set $s_u=s^{*}$ which has been fixed. So we have $(m,s) = (m', s^{*})<_{1} (m', s')$ or $(m,s) <_{2} (m', s^{*})<_{1} (m', s')$.

Case 4, $A_{<_{3}}$ is uncountable. For any $(m, u.s)\in A_{<_{3}}$, $(m, u.s)<_{3} (m', s')$, then $s = ts'$ for some $t\in \mathbb{W^{*}}$. If $t = \varepsilon$, then $u\leq m'$ for all $u\in W$, which will result in a contradiction as $\downarrow\!m'$ is countable, but $W$ is uncountable. Thus $t \neq \varepsilon$ and $\min (ut)\leq m'$ for any $u\in W$. The fact that $\min t\in \mathbb{W}$ is a fixed countable ordinal guarantees the existence of $u_{0}$ with $u_0\geq \min t$. It follows that $\min t\leq m'$ according to $(m, u_{0}.s)<_{3} (m', s')$. Hence, it turns out that $(m, s) = (m, ts')<_{3} (m', s')$.

Case 5, $A_{<_{3};<_{1}}$ is uncountable. For any $(m, u.s)\in A_{<_{3};<_{1}}$, $(m, u.s)<_{3};<_{1} (m', s')$, that is, there exists $(m_{u}, s_{u})\in \mathcal{Z}$ such that $(m, u.s)<_{3} (m_{u}, s_{u})<_{1} (m', s')$ for any $u\in W$. Then $m_{u} = m'$ and $s_{u}$ can be fixed because all of these $s_{u}$ have the same lengths as $s'$, and $s_{u}\subseteq s$. Set $s_u=s^{*}$ which has been fixed. Thus $(m, u.s)<_{3} (m', s^{*})<_{1} (m', s')$. Similar to the proof of Case $4$, we can get that $(m, s)<_{3} (m', s^{*})<_{1} (m', s')$.

$\mathbf{Claim~3}$: If $W$ is countably  infinite, then $\sup_{u\in W} (m, u.s) = (m, u_{0}.s)$, where $u_{0} = \sup W$.

It is evident that $(m, u_{0}.s)$ is an upper bound of $\{(m, u.s): u\in W\}$ because $(m, u.s)<_{1} (m, u_{0}.s)$ for each $u\in W$. Assume that $(m', s')$ is another upper bound. Now we set the same set ${A_{<_{r}}}$ as which in Claim $2$. Then there must exist at least one of the five sets being the cofinal subset of $\{(m, u.s): u\in W\}$. One sees directly that $\sup A_{<_{r}}\!= \sup_{u\in W}(m, u.s)$ if $A_{<_{r}}$ is a cofinal subset of $\{(m,u.s): u\in W\}$.

Case 1, $A_{<_{1}}$ is cofinal in $\{(m, u.s): u\in W\}$. For each $(m, u.s)\in A_{<_{1}}$, $(m, u.s)<_{1} (m', s')$, this means $m = m'$ and $u< u'$ if we set $s' = u'.s$ for some $u'\in \mathbb{W}$ fixed. Then $u_{0} = \sup W = \sup\{u: (m, u.s)\in A_{<_{1}}\}\leq u'$. Thus $(m, u_{0}.s) = (m, u'.s) = (m', s')$ or $(m, u_{0}.s) <_{1} (m, u'.s) = (m', s')$.

Case 2, $A_{<_{2}}$ is cofinal in $\{(m, u.s): u\in W\}$. For each $(m, u.s)\in A_{<_{2}}$, $(m, u.s)<_{2} (m', s')$, then $m = m'$ and $s'\subseteq s$, which yields that $(m, u_{0}.s)<_{2} (m', s')$.

Case 3, $A_{<_{2};<_{1}}$ is cofinal in $\{(m, u.s): u\in W\}$. Then for each $(m, u.s)\in A_{<_{2};<_{1}}$, $(m, u.s)<_{2};<_{1} (m', s')$, that is, there is $(m_{u}, s_{u})\in \mathcal{Z}$ such that $(m, u.s)<_{2} (m_{u}, s_{u}) <_{1} (m', s')$. This implies that $m = m_{u} = m'$ and $s_{u}$ can be fixed because all of these $s_{u}$ have the same lengths as $s'$, and $s_{u}\subseteq s$. Set $s_u=s^{*}$ which has been fixed. Hence, $(m, u_{0}.s)<_{2} (m', s^*) <_{1} (m', s')$ holds.

Case 4, $A_{<_{3}}$ is cofinal in $\{(m, u.s): u\in W\}$. Then for each $(m, u.s)\in A_{<_{3}}$, $(m, u.s)<_{3} (m', s')$, so $s = ts'$ for some $t\in \mathbb{W}^{*}$ and $\min (ut)\leq m'$. If $t = \varepsilon$ or $u\leq\min t$ for all $u$, then $\min (ut)=u\leq m'$ for each $u\in W$. This manifests that $u_{0}\leq m'$. Note that $u_0.s=u_0.ts'$ and $\min(u_0.t)=u_0\leq m'$. As a result, $(m, u_{0}.s)<_{3} (m', s')$. If $t\neq \varepsilon$ and there is a $u_{1}> \min t$ for some $u_{1}\in \{u\in W: (m, u.s)\in A_{<_{3}}\}$, then $\min t\leq m'$ because of $(m, u_{1}.s)<_{3} (m', s')$. It follows that $\min (u_{0}t) = \min t$. So we can gain that $(m, u_{0}.s)<_{3} (m', s')$.

Case 5, $A_{<_{3};<_{1}}$ is cofinal in $\{(m, u.s): u\in W\}$. For each $(m, u.s)\in A_{<_{3};<_{1}}$, we know $(m, u.s)<_{3};<_{1} (m', s')$. This means that there exists $(m_{u}, s_{u})\in \mathcal{Z}$ satisfying $(m, u.s)<_{3} (m_{u}, s_{u})$ $<_{1} (m', s')$ for each $u\in \{u\in W: (m, u.s)\in A_{<_{3};<_{1}}\}$. So $m_{u} = m'$ and $s_{u}$ can be fixed as $s^{*}$ since the length of each $s_{u}$ is same as $s'$, and $s_{u}\subseteq s$. This suggests that $(m, u.s)<_{3} (m', s^{*})<_{1} (m', s')$. By using similar analysis of Case $4$, we can obtain that $(m, u_{0}.s)<_{3} (m', s^{*})<_{1} (m', s')$.

This covers all cases to be considered and we conclude that $\mathcal{Z}$ is a $\mathrm{dcpo}$.
\end{proof}
\end{proposition}
\subsection{The Scott closed irreducible subsets of $\mathcal{Z}$}
Next, we shall characterize the Scott closed irreducible subsets of $\mathcal{Z}$ and make a conclusion that $I\!R\!R(\mathcal{Z}) = \{\downarrow\!(m,s): (m,s)\in \mathcal{Z}\}\bigcup\; \{\downarrow\!L_{s}: s\in \mathbb{W}^{*}\}$, where $L_s=\{(m,s)\in \mathcal{Z}:m\in \mathbb{W}\}$.

The  following lemma is immediate and we omit the proof. 
\begin{lemma}\label{chain}
Let $A$ be a subset of $\mathcal{Z}$. Then $A$ is Scott closed if it is downward closed and contains the sups of every non-trivial chain of the form $(m, u.s)_{u\in W}$ contained in $A$, where $W\subseteq \mathbb{W}$.
\end{lemma}

\begin{lemma}\label{unum}
Let $A\subseteq \mathcal{Z}$ be a Scott closed subset. If $A$ contains uncountable number of elements of $L_{s}$ for some $s\in \mathbb{W}^{*}$, then $L_{s}\subseteq A$.
\begin{proof}
Let $\{(m, s): m\in W\}$ be a subset of $A$ with an uncountable set $W$. For any $m\in W$, $L_{m.s}\subseteq A$ since each element of $L_{m.s}$ is related to $(m, s)$ by the relation $<_{3}$ and $A$ is a downward closed set. Thus for any $a\in \mathbb{W}$, $\sup_{m\in W}(a, m.s) = (a, s)\in A$, which implies that $L_{s}\subseteq A$.
\end{proof}
\end{lemma}

\begin{proposition}
For any $s\in \mathbb{W}^{*}$, $\downarrow\!L_{s}$ is Scott closed and irreducible.
\begin{proof}
$\mathbf{Claim~1}$: $\downarrow\!L_{s}$ is Scott closed.

Obviously, $\da L_s$ is a lower set. Now it remains to confirm that $\da L_s$ contains the supremum of every non-trivial chain of the form $(m, u.s)_{u\in W}$ contained in $\da L_s$ by Lemma \ref{chain}. To this end, let $(m, u.s_{0})_{u\in W}\subseteq\ \downarrow\!L_{s}$ be a non-trivial chain. We set
\begin{center}
$A_{<_{r}} = \{(m, u.s_{0}): u\in W,\ \exists\ (m_{u}, s)\in L_{s}\ s.t.\ (m, u.s_{0})<_{r} (m_{u}, s)\}$,
\end{center}
where $<_{r} \in\!\{<_{1},\ <_{2},\ <_{3},\ <_{2};<_{1},\ <_{3};<_{1}\}$. Then there exists at least one of the five sets being the cofinal subset of $\{(m, u.s_{0}): u\in W\}$ whether $W$ is countable infinite or uncountable.

Case 1, $A_{<_{1}}$ is cofinal in $\{(m, u.s_{0}): u\in W\}$. For any $(m, u.s_{0})\in A_{<_{1}}$, $(m, u.s_{0})<_{1} (m_{u}, s)$ for some $m_{u}\in \mathbb{W}$. Then we have $m_{u} = m, u< u'$ if we set $s$ as $u'.s_{0}$ for some $u'\in \mathbb{W}$. Under this condition, $W$ is countable and $\sup W\leq u'$. So $\sup_{u\in W} (m, u.s_{0})=(m,\sup W.s) = (m, u'.s_{0}) = (m, s)$ or $\sup_{u\in W} (m, u.s_{0})=(m,\sup W.s) <_{1} (m, u'.s_{0}) = (m, s)$.

Case 2, $A_{<_{2}}$ is cofinal in $\{(m, u.s_{0}): u\in W\}$. For any $(m, u.s_{0})\in A_{<_{2}}$, $(m, u.s_{0})<_{2} (m_{u}, s)$ for some $m_{u}\in \mathbb{W}$. We can get that $m_{u} = m$ and $s\subseteq s_{0}$, so $\sup_{u\in W} (m, u.s_{0}) = (m, s)$ or $\sup_{u\in W} (m, u.s_{0})<_{2} (m, s)$.

Case 3, $A_{<_{2};<_{1}}$ is cofinal in $\{(m, u.s_{0}): u\in W\}$. For any $(m, u.s_{0})\in A_{<_{2};<_{1}}$, there is $(m_{u}, s_{u})\in \mathcal{Z}$ such that $(m, u.s_{0})<_{2}(m_{u}, s_{u})<_{1} (m_{u}, s)$. Then we have $m_{u} = m$ , $s_{u}\subseteq s_{0}$ and $s_u$ has the same length as $s$, which leads to the result that $s_{u}$ can be fixed as $s^{*}$. Thus $\sup_{u\in W} (m, u.s_{0}) = (m, s^{*})<_{1} (m,s)$ or
$\sup_{u\in W} (m, u.s_{0}) <_{2} (m, s^{*})<_{1} (m,s)$.

Case 4, $A_{<_{3}}$ is cofinal in $\{(m, u.s_{0}): u\in W\}$. For any $(m, u.s_{0})\in A_{<_{3}}$, $(m, u.s_{0})<_{3} (m_{u}, s)$ for some $m_{u}\in \mathbb{W}$. Then $s\subseteq s_{0}$ and hence, $\sup_{u\in W}(m, u.s_{0}) = (m, s)$ or $\sup_{u\in W}(m, u.s_{0})<_{2} (m, s)$.

Case 5, $A_{<_{3};<_{1}}$ is cofinal in $\{(m, u.s_{0}): u\in W\}$. For any $(m, u.s_{0})\in A_{<_{3};<_{1}}$, $(m, u.s_{0})<_{3};<_{1} (m_{u}, s)$ for some $m_{u}\in \mathbb{W}$, which implies that there exists $s_{u}\in \mathbb{W}^{*}$ such that $(m, u.s_{0})<_{3} (m_{u}, s_{u})<_{1} (m_{u}, s)$. Then $s_{u}\subseteq s_{0}$ and $s_u$ has the same length as $s$. This yields that $s_u$ can be fixed as $s^{*}$. It follows that $\sup_{u\in W}(m, u.s_{0}) = (m, s^{*})<_{1} (m, s)$ or $\sup_{u\in W}(m, u.s_{0})<_{2} (m, s^{*})<_{1} (m, s)$.

Now we can gain our desired result that $\da L_s$ is Scott closed.

$\mathbf{Claim~2}$: $\downarrow\!L_{s}$ is irreducible.

It suffices to deduce that $L_{s}$ is irreducible. Let $U_{1}, U_{2}$ be two Scott open subsets of $\mathcal{Z}$ with $L_{s}\cap U_{1}\neq \emptyset,\ L_{s}\cap U_{2}\neq \emptyset$. Then choose $(m_{1}, s)\in L_{s}\cap U_{1}$ and $(m_{2}, s)\in L_{s}\cap U_{2}$. Because $(m_{1}, s) = \sup_{u\in \mathbb{W}}(m_{1}, u.s),\ (m_{2}, s) = \sup_{u\in \mathbb{W}}(m_{2}, u.s)$, we can find $(m_{1}, u_{1}.s)\in U_{1}$ and $(m_{2}, u_{2}.s)\in U_{2}$ for some $u_{i}\in \mathbb{W},\ i = 1,2$ from the Scott openness of $U_1$ and $U_2$. Without loss of generality, suppose that $u_{1}\leq u_{2}$, then $(m, u_{i}.s)<_{3} (u_{2}, s),\ i = 1,2$, which result in the conclusion that $(u_{2}, s)\in L_{s}\cap U_{1}\cap U_{2}$. Therefore, $L_{s}$ is irreducible.
\end{proof}
\end{proposition}
From the above proposition, we know that $\da L_{s}$ must be Scott irreducible closed sets. Conversely, to prove that $IRR(\mathcal{Z})\subseteq \{\da L_s : s\in \mathbb{W}^*\}$, we need to consider $\mathbb{W}^{*}$ endowed with the order $\sqsubseteq$ at first, which is defined in the following.

Based on the observation of the definition of the order relation on $\mathcal{Z}$, we can find that the two relations $<_{1}, <_{2}$ concern only the strings component of points $(m, s)$.  We use them to define an order $\sqsubseteq$ on $\mathbb{W}^{*}$(where $u, u'\in \mathbb{W}$, $s, t\in \mathbb{W}^{*}$):
\begin{itemize}
\item $u.s\sqsubset_{1} u'.s$ if $u< u'$
\item $ts\sqsubset_{2} s$ if $t\neq \varepsilon$.
\end{itemize}
Then $\sqsubset\ :=\ \sqsubset_{1}\cup \sqsubset_{2}\cup\ (\sqsubset_{2};\sqsubset_{1})$ is transitive and irreflexive, hence, $\sqsubseteq\ := (\sqsubset\cup =)$ is an order relation on $\mathbb{W}^{*}$, and $\mathbb{W}^{*}$ with the order relation $\sqsubseteq$ is a poset. We denote it by $\mathcal{T}$ and describe parts of $\mathcal{T}$ in Figure $2$.

\begin{remark}\label{lower}
Combine with the order on $\mathcal{Z}$, we have:
\begin{enumerate}
\item If $s\sqsubseteq s'$, then $(m, s)\leq (m, s')$;
\item If $(m, s)\leq (m', s')$, then $s\sqsubseteq s'$,
\end{enumerate}
where $m, m'\in \mathbb{W}$, $s, s'\in \mathbb{W}^{*}$.
\end{remark}

We state and derive an analogous result for $\mathcal{T}$ with the aid of \cite[Proposition 5.9]{hzp} and Proposition~\ref{dcpo}.
\begin{proposition}\label{sf}
\
\begin{enumerate}
\item For any $s\in \mathbb{W}^{*}$, $\uparrow\!s$ is a linear ordered subset of $\mathcal{T}$.
\item The poset $\mathcal{T}$ endowed with the Scott topology is sober.
\item The map $f: \mathcal{Z}\rightarrow \mathcal{T}$ defined by $f(m, s) = s$ is Scott continuous.
\end{enumerate}
\end{proposition}

\begin{figure}[H]
\centering
\begin{tikzpicture}[line width=0.6pt,scale=0.7]
\draw (0,0) circle (2pt) ;
\draw (0,0.8) circle (2pt);
\draw (0,1.6) circle (2pt);
\draw (0,2.6) circle (2pt);
\draw (0,3.4) circle (2pt);
\draw (0,5) circle (2pt);

\draw (-1.4,-0.7) circle (2pt);
\draw (-2,-1) circle (2pt);
\draw (-2.8,-1.4) circle (2pt);
\draw (-3.4,-1.7) circle (2pt);
\draw (-4,-2) circle (2pt);

\draw (1.8,0.7) circle (2pt);
\draw (2.4,0.4) circle (2pt);
\draw (3.2,0) circle (2pt);
\draw (3.8,-0.3) circle (2pt);
\draw (4.4,-0.6) circle (2pt);

\draw (3.2,-1.6) circle (2pt);
\draw (3.2,-2.4) circle (2pt);
\draw (3.2,-3.4) circle (2pt);
\draw (3.2,-4.2) circle (2pt);
\draw (3.2,-5) circle (2pt);

\draw (0,0.05)--(0,0.75) (0,0.85)--(0,1.55) (0,2.65)--(0,3.35) (0,3.45)--(0,3.95);
\draw (-3.97,-1.985)--(-3.435,-1.7175) (-3.37,-1.685)-- (-2.83,-1.415) (-1.37,-0.685)--(-0.85,-0.425) (-1.97,-0.985)--(-1.43,-0.715);
\draw [dashed] (0,1.65)--(0,2.55) (0,4)--(0,4.95) (-2.77,-1.385)--(-2.03, -1.015) (-0.73,-0.365)--(0,0) ;

\draw (1.84,0.68)--(2.38,0.41) (3.23,-0.015)--(3.77,-0.285) (3.83,-0.315)--(4.37,-0.585);

\draw (3.2,-4.15)--(3.2,-3.45) (3.2,-4.95)--(3.2,-4.25) (3.2,-2.35)--(3.2,-1.65) (3.2,-1.55)--(3.2,-1) ;
\draw [dashed] (3.2,-3.35)--(3.2,-2.45) (3.2,-1)--(3.2,-0.05);
\draw [dashed] (0.1,1.55)--(1.7,0.75) (2.5, 0.35)--(3.1,0.05);
\draw [dashed] (0.1,0.75)--(1,0.3) (0.1,2.55)--(1,2.1) (0.1, 3.35)--(1,2.9) (1.8,0.65)--(1.8,-0.7)  (2.4,0.35)--(2.4,-1)(3.8,-0.35)--(3.8,-1.6) (4.4,-0.65)--(4.4,-1.95);
\draw [dashed] (-1.4,-0.8)--(-1.4,-2)  (-2,-1.1)--(-2,-2.3) (-2.8,-1.5)--(-2.8,-2.7) (-3.4,-1.8)--(-3.4,-3) (-4,-2.1)--(-4,-3.3);
\draw [dashed] (3.1,-1.65)--(2.1,-2.15) (3.1,-2.45)--(2.1,-2.95) (3.1,-3.45)--(2.1,-3.95) (3.1,-4.25)--(2.1,-4.75) (3.1,-5.05)--(2.1,-5.55);

\node[right][font=\scriptsize]  at(0,0) {1};
\node[left][font=\scriptsize]  at(0,0.8) {2};
\node[left][font=\scriptsize]  at(0,1.6) {3};
\node[left][font=\scriptsize]  at(0,2.6) {$\omega_{0}$};
\node[left][font=\scriptsize]  at(0,3.4) {$\omega_{0}\!+\!1$};
\node[left][font=\scriptsize]  at(0,5) {$\varepsilon$};

\node[left][font=\scriptsize]  at(-0.9,-0.3) {$(\!\omega_{0}\!\!+\!\!1\!).1$};
\node[left][font=\scriptsize]  at(-1.9,-0.8) {$\omega_{0}.1$};
\node[left][font=\scriptsize]  at(-2.7,-1.2)  {3.1};
\node[left][font=\scriptsize]  at(-3.3,-1.5)  {2.1};
\node[left][font=\scriptsize]  at(-4,-2)  {1.1};

\node[right][font=\scriptsize]  at(1.3,1.1) {$(\!\omega_{0}\!\!+\!\!1\!).3$};
\node[right][font=\scriptsize]  at(2.4,0.6) {$\omega_{0}.3$};
\node[right][font=\scriptsize]  at(3.25,0.2) {3.3};
\node[right][font=\scriptsize]  at(3.8,-0.2) {2.3};
\node[right][font=\scriptsize]  at(4.4,-0.6) {1.3};

\node[left][font=\scriptsize]  at(3.2,-1.5) {$(\!\omega_{0}\!\!+\!\!1\!).3.3$};
\node[right][font=\scriptsize]  at(3.2,-2.4) {$\omega_{0}.3.3$};
\node[right][font=\scriptsize]  at(3.2,-3.4) {3.3.3};
\node[right][font=\scriptsize]  at(3.2,-4.2) {2.3.3};
\node[right][font=\scriptsize]  at(3.2,-5)   {1.3.3};

\end{tikzpicture}
\scriptsize \\ Figure 2: The order $\sqsubseteq$ on $\mathbb{W}^{*}$.
\end{figure}

We remark that for any chain $C\subseteq \mathbb{W}^{*}$ (or $C'\subseteq \mathcal{Z}$), $\sup C\ (\sup C')$ exists and $\sup C = \sup_{u\in W_{1}}u.s$ (or $\sup C' = \sup_{u\in W_{2}}(m, u.s)$) since there exists $\{u.s: u\in W_{1}\}$ (or $\{(m, u.s): u\in W_{2}\}$) being a cofinal subset of $C$ (or $C'$), where $W_{1}$ (or $W_{2}$) is the subset of $\mathbb{W}$. Thus in the following context, we can regard any non-trivial chain in $\mathcal{T}$ (or in $\mathcal{Z}$) as the form of $\{u.s: u\in W_{0}\}$ (or $\{(m, u.s): u\in W'\}$) with some $W_{0}$ (or $W'$)$\subseteq \mathbb{W}$. In addition, a straightforward verification establishes that $\sup_{u\in W_0}u.s=s$ when $W_0$ is uncountable, and $\sup_{u\in W_0}u.s=\sup W_0.s$ when $W_0$ is countable.
\begin{lemma}
Let $B$ be a subset of\ \ $\mathbb{W}^{*}$. If $B$ is closed for the $\mathrm{sups}$ of any uncountable non-trivial chain, then we have $\overline{B} = B'$, where
\begin{center}
$B' = \bigcup\; \{\downarrow\!\sup C: C\subseteq\ \downarrow\!B\ \mathrm{is}\ \mathrm{a}\ \mathrm{countable}\ \mathrm{chain}\}$.
\end{center}
\begin{proof}
One sees immediately that $B'\subseteq \overline{B}$ and $B\subseteq B'$, then it remains to prove that $B'$ is Scott closed. 

$\textbf{Claim~1}$: $\downarrow\!B$ is closed for the $\mathrm{sups}$ of any uncountable chain contained in $\downarrow\!B$. 

Assume that $\{u.s': u\in W\}\subseteq\; \downarrow\!B$ be an uncountable chain. Then there exists $s_{u}\in B$ such that $u.s'\sqsubseteq s_{u}$ for each $u\in W$. If for any $u\in W$, $u.s' = s_{u}$ or $u.s'\sqsubset_{1} s_{u}$, then $\{s_{u}: u\in W\}$ is an uncountable chain of $B$, so $\sup_{u\in W} s_{u}\in B$ by the assumption that $B$ is closed for the $\mathrm{sups}$ of any uncountable non-trivial chain. This yields that $\sup_{u\in W}u.s'=s'\sqsubseteq \sup_{u\in W} s_{u}$. If there exists $u'.s'\sqsubset_{2} s_{u'}$, then $u.s'\sqsubset_{2} s_{u'}$ for all $u\in W$. Thus $s' = s_{u'}$ or $s'\sqsubset_{2} s_{u'}$. The case of which there exists $u'.s'\sqsubset_{2};\sqsubset_{1} s_{u'}$ is similar to the former. So $\sup_{u\in W}u.s'=s'\in \downarrow\!B$. 

$\textbf{Claim~2}$: $B'$ is Scott closed.

Let $\{v.s: v\in W_{1}\}$ be a non-trivial chain of $B'$ with $W_{1}\subseteq \mathbb{W}$ being uncountable. Then there exists $C_{v}\subseteq B$ being a countable chain such that $v.s\sqsubseteq \sup C_{v}$ for any $v\in W_{1}$. We distinguish the following two cases:

Case 1, For each $v\in W_{1}$, $v.s = \sup C_{v}$ or $v.s \sqsubset_{1} \sup C_{v}$. If $\{\sup C_{v}: v\in W_{1}\}$ is uncountable, then there must exist $W_{2}\subseteq \mathbb{W}$ being uncountable such that $\{a.s: a\in W_{2}\}$ is a cofinal subset of $\bigcup_{v\in W_{1}}C_{v}$. By Claim $1$, we have $s=\sup_{a\in W_{2}}a.s\in\ \downarrow\!B$. Thus $s\in B'$. If $\{\sup C_{v}: v\in W_{1}\}$ is countable, then there exists a $v'\in W_{1}$ such that the set $\{v.s: v\in W_{1}, v.s\sqsubset_{1} \sup C_{v'}\}$ is uncountable, which implies that $\{v.s: v\in W_{1}, v.s\sqsubset_{1} \sup C_{v'}\}$ is a cofinal subset of $\{v.s: v\in W_{1}\}$. It follows that $v\leq v_0$ if we set $\sup C_{v'}$ as $v_0.s$ for some $v_0\in \mathbb{W}$. This means that $W_1\subseteq \{a\in \mathbb{W}:a\leq v_0\}$. So we can obtain that $W_1$ is a countable set, which is impossible.

Case 2, There exists $v_{0}\in W_{1}$ such that $v_{0}.s\sqsubset_{2} \sup C_{v_{0}}$ or $v_{0}.s\sqsubset_{2};\sqsubset_{1} \sup C_{v_{0}}$ (since the proof of the two cases are similar, we omit that of the latter). Then we have $v.s \sqsubset_{2} \sup C_{v_{0}}$ for all $v\in W_{1}$. Hence, $\sup _{v\in W_1}v.s=s = \sup C_{v_{0}}$ or $\sup _{v\in W_1}v.s=s \sqsubset_{2} \sup C_{v_{0}}$, so $s\in B'$.

So we have proved that $B'$ is closed for the $\mathrm{sups}$ of any uncountable chain. Ones sees obviously that $B'$ is closed for the $\mathrm{sups}$ of any countable chain and $B'$ is a lower set. We can conclude that $B'$ is Scott closed. 
\end{proof}
\end{lemma}

\begin{lemma}\label{supin}
Let $A$ be a Scott closed irreducible subset of $\mathcal{Z}$. Then $\sup f(A)$ exists and $\sup f(A)\in f(A)$, where $f$ is the Scott continuous map defined in Proposition \ref{sf}.
\begin{proof}
By reason that $A$ is irreducible in $\mathcal{Z}$ and $f$ is Scott continuous, we have $\overline{f(A)}\in I\!R\!R(\mathcal{T})$. Then there exists $s_{0}\in \mathbb{W}^{*}$ such that $\overline{f(A)} = \overline{\{s_{0}\}} =\ \downarrow\!\!s_{0}$ from the sobriety of $\mathcal{T}$, which results in $\sup f(A) = s_{0}$. It suffices to show that $s_{0}\in f(A)$.

$\mathbf{Claim~1}$: $f(A)$ is closed for the $\mathrm{sups}$ of any uncountable non-trivial chain, and a lower set.

Due to the item (2) in Remark \ref{lower}, the result that $f(A)$ is a lower set follows immediately. Let $\{u.s: u\in W\}$ be an uncountable non-trivial chain in $f(A)$. If for any $u\in W$, $A_{u.s}$ is uncountable, where $A_{u.s} = \{m\in \mathbb{W}: (m, u.s)\in A\}$, then $L_{u.s}\subseteq A$ by Lemma \ref{unum}, so for each $m\in \mathbb{W}, \{(m, u.s): u\in W\}\subseteq A$. Because $W$ is uncountable and $A$ is Scott closed, $\sup_{u\in W}(m, u.s) =(m,s)\in A$ for any $m\in \mathbb{W}$. It follows that $s\in f(A)$. If $A_{u'.s}$ is countable for some $u'\in W$, then the fact that $A$ is a lower set implies that $A_{u.s}\subseteq A_{u'.s}$ for any $u\geq u'$. This means that $A_{u.s}$ is countable for any $u\geq u'$, which guarantees the existence of $u_{1}\geq u'$ with the property $A_{u.s} = A_{u_{1}.s}$ for all $u\geq u_{1}$, owing to the uncountability of $W$. Thus for each $m\in A_{u_{1}.s}, \{(m, u.s): u\geq u_{1}\}\subseteq A$. Therefore, we have $\sup_{u\geq u_1}(m, u.s)=(m, s)\in A$ for any $m\in A_{u_{1}.s}$ as $A$ is Scott closed, that is, $s\in f(A)$.

Assume for the sake of a contradiction that $s_{0}\notin f(A)$. By Claim 1, we know $\overline{f(A)} = f(A)'$. This infers that $s_{0}\in \overline{f(A)}\setminus f(A) = f(A)'\setminus f(A)$, where
\begin{center}
$f(A)' = \bigcup \{\downarrow\!\sup C: C\subseteq\; \downarrow\!f(A) = f(A)\ \mathrm{is}\ \mathrm{a}\ \mathrm{countable}\ \mathrm{chain}\}$.
\end{center}
This implies that there exists a countable chain $C\subseteq f(A)$ satisfying $s_{0} = \sup C$. Assume that $\sup C = \sup_{u\in D}u.s'=\sup D.s'$, where $D\subseteq \mathbb{W}$ is a countable subset, that is, $s_0=\sup D.s'$. Then there must exist $u'\in D$ such that $A_{u'.s'}$ is countable. If not, $A_{u.s'}$ is uncountable for all $u\in D$, then $L_{u.s'}\subseteq A$ by Lemma \ref{unum}, and for each $m\in \mathbb{W}$, $\{(m, u.s'): u\in D\}\subseteq A$, which reveals that $(m, s_{0}) = (m, \sup_{u\in D}u.s')\in A$, that is, $s_{0}\in f(A)$, a contradiction. Thus $A_{u'.s'}$ is countable for some $u'\in D$. This guarantees the existence of $\sup A_{u'.s'}$. Now we let $m_{0} = \sup A_{u'.s'}$, $B = \{s\in f(A): s\geq u'.s'\}$. Set
\begin{center}
$P = \bigcup_{s\in B}\bigcup_{m\in A_{s}}\downarrow\!(m, s)\bigcup \downarrow\!(m_{0}, s_{0})$.
\end{center}

$\mathbf{Claim~2}$: $P$ is Scott closed.

Let $\{(a, u.s): u\in D'\}$ be a non-trivial chain in $P$. If $\{(a, u.s): u\in D'\}\bigcap \downarrow\!(m_{0}, s_{0})$ is cofinal in $\{(a, u.s): u\in D'\}$, then obviously, $\sup_{u\in D'}(a, u.s)$ is less than or equal to $(m_{0}, s_{0})$. Therefore, it belongs to $P$. If $\{(a, u.s): u\in D'\}\bigcap\, (\bigcup_{s\in B}\bigcup_{m\in A_{s}}\downarrow\!(m, s))$ is cofinal in $\{(a, u.s): u\in D'\}$, then for any $u\in D'$, there exist $s_{u}\in B$ and $m_{u}\in A_{s_{u}}$ such that $(a, u.s)< (m_{u}, s_{u})$. Let
\begin{center}
$A_{<_{r}} = \{(a, u.s): u\in D'\ \mathrm{and}\ (a, u.s)<_{r} (m_{u}, s_{u})\ \mathrm{for}\ \mathrm{some}\ s_{u}\in B,\ m_{u}\in A_{s_{u}}\}$,
\end{center}
and
\begin{center}
$U_{<_{r}} = \{u\in D': (a, u.s)\in A_{<_{r}}\}$,
\end{center}
where $<_{r}\in \{<_{1},\; <_{2},\; <_{3},\; <_{2};<_{1},\; <_{3};<_{1}\}$. Now we need to distinguish the following cases.

Case 1, $A_{<_{1}}$ is cofinal in $\{(a, u.s): u\in D'\}$. For any $(a, u.s)\in A_{<_{1}}$, there exist $s_{u}\in B$ and $m_{u}\in A_{s_{u}}$ such that $(a, u.s)<_{1} (m_{u}, s_{u})$. Then $m_{u} = a$ and $\{s_{u}: u\in U_{<_{1}}\}$ is a chain contained in $B$, so $\sup_{u\in U_{<_{1}}}s_{u}$ exists and $(a, \sup_{u\in U_{<_{1}}}s_{u})\in A$. This means that $\sup_{u\in U_{<_{1}}}s_{u}\in B$ and $\sup_{u\in D'}(a, u.s)\leq (a, \sup_{u\in U_{<_{1}}}s_{u})$. We conclude that $\sup_{u\in D'}(a, u.s)\in P$.

Case 2, $A_{<_{2}}$ or $A_{<_{2};<_{1}}$ is cofinal in $\{(a, u.s): u\in D'\}$. It is easy to verify that $\sup_{u\in D'}(a, u.s)\in P$ , so we omit the proof here.

Case 3, $A_{<_{3}}$ is cofinal in $\{(a, u.s): u\in D'\}$. For any $(a, u.s)\in A_{<_{3}}$, there exist $s_{u}\in B$ and $m_{u}\in A_{s_{u}}$ such that $(a, u.s)<_{3} (m_{u}, s_{u})$, so $s_{u}\subseteq s$. The finiteness of the length of $s$ ensures the existence of $s_{u_{0}}\in B$ such that $\{(a, u.s): u\in U_{<_{3}}, (a, u.s)<_{3} (m_{u}, s_{u_{0}})\}$ is cofinal in $A_{<_{3}}$. Assume that $s = ts_{u_{0}}$, then $\min(ut)\leq m_{u}$. If $t = \varepsilon$ or $u\leq \min t$ for all $u\in U_{<_{3}}$, then $u\leq m_{u}$ for any $u\in U_{<_{3}}$ and $D'$ is countable under this condition.  It follows that $A_{s_{u_{0}}}\subseteq A_{u'.s'}$ as $s_{u_{0}}\sqsupseteq u'.s'$. This indicates that $\sup D' = \sup U_{<_{3}}\leq\sup_{u\in U_{<_{3}}}m_{u}\leq m_{0}$. Thus that $\sup_{u\in D'}(a, u.s)<_{3} (m_{0}, s_{u_{0}})< (m_{0}, s_{0})$ holds. Else, $t\neq \varepsilon$ and there is a $u_{1}\in D'$ such that $u_{1}\!> \min t$. Then we have $\min t\leq m_{u_1}$ according to $(a, u_1.s)<_{3} (m_{u_1}, s_{u_{0}})$, which leads that $\sup_{u\in D'}(a, u.s)<_{3} (m_{u_{1}}, s_{u_{0}})$. Hence $\sup_{u\in D'}(a, u.s)\in P$.

Case 4, $A_{<_{3};<_{1}}$ is cofinal in $\{(a, u.s): u\in D'\}$. For any $(a, u.s)\in A_{<_{3};<_{1}}$, there exist $s_{u}\in B$ and $m_{u}\in A_{s_{u}}$ such that $(a, u.s)<_{3};<_{1} (m_{u}, s_{u})$, that is, $(a, u.s)<_{3} (m_{u}, s_{u}')<_{1} (m_{u}, s_{u})$ for some $s_{u}'\in \mathbb{W}^{*}$. Then there exists $s^{*}\in \mathbb{W}^{*}$ such that $\{(a, u.s): u\in U_{<_{3};<_{1}}\ \mathrm{and}\ (a, u.s)<_{3} (m_{u}, s^{*})<_{1} (m_{u}, s_{u})\}$ is cofinal in $A_{<_{3};<_{1}}$ as the length of $s$ is finite and $s_{u}'\subseteq s$. The rest of the proof is similar to Case $3$. Thus we have $\sup_{u\in D'}(a, u.s)\in P$.

Therefore, $P$ is Scott closed.

Note that $A\subseteq P\cup \overline{(A\setminus\!P)}$ and $A$ is irreducible. Then $A\subseteq P$ or $A\subseteq \overline{A\setminus\!P}$. 

$\textbf{Claim~3}$: $A\subseteq P$.

Suppose $A\not \subseteq P$. Then $A\subseteq \overline{A\setminus\!P}$. The continuity of $f$ results in $\overline{f(A)}\subseteq \overline{f(A\setminus\!P)}$. It follows that $\overline{f(A)}\subseteq f(A\setminus\!P)'$ as $f(A\setminus\!P)$ is closed for the sups of any uncountable non-trivial chain (the proof is similar to that of $f(A)$). Then $s_{0}\in f(A\setminus\!P)'$. This means $s_{0} = \sup C_{0}$ for some countable chain $C_{0}\subseteq\ \downarrow\!\!f(A\setminus\!P)$. Since $s_{0} = \sup_{u\in D}u.s'$, we can obtain that $C_0$ is a cofinal subset of $\{u.s'\mid u\in D\}$, which yields that there exists a $s''\in f(A\setminus\!P)$ such that $u'.s'\sqsubseteq s''$. So we can find $m'\in \mathbb{W}$ such that $(m', s'')\in A\setminus\!P$. However, according to the construction of $P$, the result that $(m', s'')\in P$ follows immediately, which is a contradiction to the fact that $(m', s'')\in A\backslash P$. Hence, $A\subseteq P$.

Now we claim that for any $u\in D$, there is a $u^{*}\geq u$ such that $A_{u^{*}.s'} \subsetneqq A_{u.s'}$.
Suppose not, if there is a $u^{*}\in D$ such that for any $u\in D$ with $u\geq u^{*}$, $A_{u.s'} = A_{u^{*}.s'}$, then for each $m\in A_{u^{*}.s'}$, $\{(m, u.s'): u\in D, u\geq u^{*}\}\subseteq A$, and we know that $(m, s_{0}) = (m, \sup_{u\in \uparrow\!u^{*}\cap D}u.s')\in A$, which contradicts the assumption that $s_{0}\notin f(A)$. 

Thus for any $u\in D$, there is a $u^{*}\geq u$ such that $A_{u^{*}.s'} \subsetneqq A_{u.s'}$. It follows that for $u'\in D$, we can find $u''\in D$ such that $u''> u'$, $A_{u''.s'} \subsetneqq A_{u'.s'}$. The assumption that $s_0\notin f(A)$ reveals that $m_{0} \notin  A_{u_{1}.s'}$ for some $u_1> u''$. 
It turns out that there is $u_2>u_1$ in $D$ satisfying that $A_{u_2.s'}\subsetneqq A_{u_1.s'}$.
Write $E = \{u\in D: A_{u.s'} = A_{u_{1}.s'}\}$. The Scott closedness of $A$ indicates that $\sup E\in E$. Let $\hat{u} = \sup E$. Then $A_{\hat{u}+1.s'}\subsetneqq A_{\hat{u}.s'}$ as $\hat{u}+1\notin E$ and $\hat{u}+1\leq u_2$, which yields $A_{u_2.s'}\subseteq A_{\hat{u}+1.s'}$. Note that $A_{u_2.s'}$ is nonempty by reason of $u_2.s'\in f(A)$. Then $A_{\hat{u}+1.s'}$ is nonempty. Pick $m_{1}\in A_{\hat{u}.s'}\setminus A_{\hat{u}+1.s'}$. We can get that $m_{0}\neq m_{1}$ from $A_{\hat{u}.s'} = A_{u_{1}.s'}$.

Now we set $Q = \downarrow\!(P\setminus \downarrow\!(m_{1}, \hat{u}.s'))$, then $P = \downarrow\!(m_{1}, \hat{u}.s')\cup Q$.

$\mathbf{Claim~4}$: $Q$ is Scott closed.

By the construction of $P$, we know
\begin{center}
$Q = \downarrow (\bigcup_{s\in B}\bigcup_{m\in A_{s}}\downarrow\!(m,s)\setminus \downarrow\!(m_{1}, \hat{u}.s'))\bigcup \downarrow (\downarrow\!(m_{0}, s_{0})\setminus \downarrow\!(m_{1}, \hat{u}.s'))$.
\end{center}
The set $Q$ is obviously a lower set. So it remains to prove that $Q$ is closed for the $\mathrm{sups}$ of any non-trivial chain.
To begin with, we first notice that $(m_0,s_0)\notin  \downarrow\!(m_{1}, \hat{u}.s')$. Suppose $(m_{0}, s_{0})\leq (m_{1}, \hat{u}.s')$, then $(m_0,s_0)\in A$ because of the fact that $ (m_{1}, \hat{u}.s')\in A$, which is a contradiction to the assumption that $s_0\notin f(A)$. It follows that $\downarrow (\downarrow\!(m_{0}, s_{0})\setminus \downarrow\!(m_{1}, \hat{u}.s'))=\da (m_0,s_0)$. In other words, 
\begin{center}
$Q = \downarrow (\bigcup_{s\in B}\bigcup_{m\in A_{s}}\downarrow\!(m,s)\setminus \downarrow\!(m_{1}, \hat{u}.s'))\bigcup \downarrow\!(m_{0}, s_{0})$.
\end{center}

Let $\{(a, u.s): u\in D_{0}\}$ be a non-trivial chain in $Q$. If $\{(a, u.s): u\in D_{0}\}\bigcap \downarrow (m_{0}, s_{0})$ is cofinal in $\{(a, u.s): u\in D_{0}\}$, then $\sup_{u\in D_{0}}(a, u.s)\leq (m_{0}, s_{0})$. Else, $\{(a, u.s): u\in D_{0}\}\bigcap\downarrow (\bigcup_{s\in B}\bigcup_{m\in A_{s}}\downarrow\!(m,s)\setminus \downarrow\!(m_{1}, \hat{u}.s'))$ is cofinal in $\{(a, u.s): u\in D_{0}\}$. Then, for any $u\in D_{0}$, there exists $s_{u}\in B, m_{u}\in A_{s_{u}}$ such that $(a, u.s)< (m_{u}, s_{u})$ and $(m_{u}, s_{u})\nleq (m_{1}, \hat{u}.s')$. We set
\begin{center}
$A_{<_{r}} = \{(a, u.s): u\in D_{0},\ \exists s_{u}\in B, m_{u}\in A_{s_{u}}\ \mathrm{with}\ (m_{u}, s_{u})\nleq (m_{1}, \hat{u}.s')\ \mathrm{s.t.}\ (a, u.s)<_{r} (m_{u}, s_{u})\}$,
\end{center}
and
\begin{center}
$U_{<_{r}} = \{u\in D_{0}: (a, u.s)\in A_{<_{r}}\}$,
\end{center}
where $<_{r}\in \{<_{1},\ <_{2},\ <_{3},\ <_{2};<_{1},\ <_{3};<_{1}\}$. Now we need to distinguish the following cases.

Case 1, $A_{<_{1}}$ is cofinal in $\{(a, u.s): u\in D_{0}\}$. For any $(a, u.s)\in A_{<_{1}}$, there exist $s_{u}\in B$ and $m_{u}\in A_{s_{u}}$ with $(m_{u}, s_{u})\nleq (m_{1}, \hat{u}.s')$ such that $(a, u.s)<_{1} (m_{u}, s_{u})$. Then we have $m_{u} = a$ and $\{s_{u}: u\in U_{<_{1}}\}$ is a chain in $B$. It follows that $(a, \sup_{u\in U_{<_{1}}}s_{u})\in A$ from the Scott closedness of $A$. Furthermore, $\sup_{u\in D_{0}}(a, u.s)\leq (a, \sup_{u\in U_{<_{1}}}s_{u})$. It is obvious that $(a, \sup_{u\in U_{<_{1}}}s_{u})\nleqslant (m_{1}, \hat{u}.s')$ since $(a, s_{u})\nleqslant (m_{1}, \hat{u}.s')$ for each $u\in U_{<_{1}}$. So we have $\sup_{u\in D_{0}}(a, u.s)\in Q$.

Case 2, $A_{<_{2}}$ or $A_{<_{2};<_{1}}$ is cofinal in $\{(a, u.s): u\in D_{0}\}$. Clearly, $\sup_{u\in D_{0}}(a, u.s)\in Q$ in this case.

Case 3, $A_{<_{3}}$ is cofinal in $\{(a, u.s): u\in D_{0}\}$. For any $(a, u.s)\in A_{<_{3}}$, there exist $s_{u}\in B$ and $m_{u}\in A_{s_{u}}$ with $(m_{u}, s_{u})\nleq (m_{1}, \hat{u}.s')$ such that $(a, u.s)<_{3} (m_{u}, s_{u})$. Then $s_{u}\subseteq s$. Since the length of $s$ is finite, there exists $s_{u_{0}}\in B$ such that $\{(a, u.s): u\in U_{<_{3}}, (a, u.s)<_{3} (m_{u}, s_{u_{0}})\}$ is cofinal in $A_{<_{3}}$. Let $s = ts_{u_{0}}$. Then $\min (ut)\leq m_{u}$ and $m_{u}\in A_{s_{u_{0}}}\subseteq A_{u'.s'}$ for all $u$. If $t = \varepsilon$ or $u\leq \min t$ for each $u$, then $u=\min(ut)\leq m_{u}$, which implies that $\sup U_{<_{3}}\leq m_{0},\sup U_{<_{3}}\leq \min t$ and $D_{0}$ is countable. So $\sup_{u\in D_{0}}(a, u.s) <_{3} (m_{0}, s_{u_{0}})\leq (m_{0}, s_{0})$. Thus $\sup_{u\in D_{0}}(a, u.s)\in Q$. If $t\neq \varepsilon$ and $u_{1}> \min t$ for some $u_{1}\in U_{<_{3}}$, then $\min t\leq m_{u_{1}}$, which yields that $\sup_{u\in D_{0}}(a, u.s)<_{3} (m_{u_{1}}, s_{u_{0}})$. This means that $\sup_{u\in D_{0}}(a, u.s)\in Q$ by the fact that $(m_{u_{1}}, s_{u_{0}})\nleq (m_{1}, \hat{u}.s')$.

Case 4, $A_{<_{3};<_{1}}$ is cofinal in $\{(a, u.s): u\in D_{0}\}$. For any $(a, u.s)\in A_{<_{3};<_{1}}$, there exist $s_{u}\in B$ and $m_{u}\in A_{s_{u}}$ with $(m_{u}, s_{u})\nleq (m_{1}, \hat{u}.s')$ such that $(a, u.s)<_{3};<_{1} (m_{u}, s_{u})$, that is, there is a $s_{u}'$ for each $u$ such that $(a, u.s)<_{3} (m_{u}, s_{u}')<_{1} (m_{u}, s_{u})$. We can also find $s^{*}\in \mathbb{W}^{*}$ to make $\{(a, u.s): u\in U_{<_{3};<_{1}}, (a, u.s)<_{3} (m_{u}, s^{*})<_{1} (m_{u}, s_{u})\}$ cofinal in $A_{<_{3};<_{1}}$. Then with a process similar to that in Case $3$, we conclude that $\sup_{u\in D_{0}}(a, u.s)\in Q$.

Therefore, $Q$ is Scott closed. Since $A\subseteq P$ is irreducible in $\mathcal{Z}$, we have $A\subseteq\ \downarrow\!(m_{1}, \hat{u}.s')$ or $A\subseteq Q =\ \downarrow\!(P\setminus \downarrow\!(m_{1}, \hat{u}.s'))$. If $A\subseteq\ \downarrow\!(m_{1}, \hat{u}.s')$, then for any $m\in A_{\hat{u}+1.s'}$, $(m,\hat{u}+1.s')\in A\subseteq\ \downarrow\!(m_{1}, \hat{u}.s')$. This means that $\hat{u}+1.s'\leq \hat{u}.s'$, a contradiction.
Thus $A\subseteq Q =\ \downarrow\!(P\setminus \downarrow\!(m_{1}, \hat{u}.s'))$. It follows that $(m_{1}, \hat{u}.s')\in\ \downarrow\!(P\setminus \downarrow\!(m_{1}, \hat{u}.s'))$, that is, there exists $(m_{2}, s_{2})\in P\setminus \downarrow\!(m_{1}, \hat{u}.s')$ such that $(m_{1}, \hat{u}.s')\leq (m_{2}, s_{2})$.
Now we need to distinguish the following two cases.

Case $1$, $(m_{2}, s_{2})\in (\bigcup_{s\in B}\bigcup_{m\in A_{s}}\downarrow\!(m,s)\setminus \downarrow\!(m_{1}, \hat{u}.s'))$. Then $(m_{2}, s_{2})\leq (m_{3},s_{3})$, where $s_{3}\in B$, $m_{3}\in A_{s_{3}}$ and $(m_3,s_3)\nleq (m_{1}, \hat{u}.s')$. Via $(m_{1}, \hat{u}.s')\leq (m_{2}, s_{2})\leq (m_{3},s_{3})$, we have $\hat{u}.s'\sqsubseteq s_{2}\sqsubseteq s_{3}\sqsubset s_{0} = \sup_{u\in D}u.s'$. So $m_{1} = m_{2} = m_{3}$. The assumption that $(m_{3}, s_{3})\nleqslant (m_{1}, \hat{u}.s')$ indicates that $(m_{1}, \hat{u}.s')<_{1} (m_{3}, s_{3})$. Then $\hat{u}.s'\sqsubset_{1} s_{3}$. This implies $(\hat{u}\!+\!1).s'\sqsubset_{1} s_{3}$ or $(\hat{u}\!+\!1).s' = s_{3}$. Hence, it turns out that $(\hat{u}\!\!+\!\!1).s'\sqsubseteq s_{3}$, which yields that $A_{s_{3}}\subseteq A_{\hat{u}+1.s'}$ and $m_{1}\in A_{\hat{u}+1.s'}$ by reason of the fact that $m_{1} = m_{3}\in A_{s_{3}}$. This is a contradiction to the assumption that $m_{1}\notin A_{\hat{u}+1.s'}$. 

Case $2$, $(m_{2}, s_{2})\in \downarrow\!(m_{0}, s_{0})\setminus \downarrow\!(m_{1}, \hat{u}.s')$. Then $(m_{1}, \hat{u}.s')< (m_{2}, s_{2})< (m_{0}, s_{0})$. This means that $\hat{u}.s'\sqsubset_1 s_0$, which leads that $m_{1} = m_{0}$. It violates the assumption that $m_1\neq m_0$.

In a conclusion, $s_{0}\in f(A)$.
\end{proof}
\end{lemma}

\begin{theorem}\label{A}
The irreducible closed subsets of $\mathcal{Z}$ are just closures of single elements and the closures of levels.
\begin{proof}
Let $A$ be an irreducible closed subset of $\mathcal{Z}$ and $s_{0} = \sup f(A)$. Then we have $s_{0}\in f(A)$ by Lemma \ref{supin}, which implies that $A_{s_{0}}\neq \emptyset$ and $f(A)\subseteq\ \downarrow\!s_{0}$, it is equal to saying that $A\subseteq\ \downarrow\!L_{s_{0}}$. Next we want to show that $A = \downarrow\!L_{s_{0}}$ or $A = \downarrow\!(m_{0}, s_{0})$ for some $m_{0}\in \mathbb{W}$. If $A_{s_{0}}$ is uncountable, then $L_{s_{0}}\subseteq A$ and so $A = \downarrow\!L_{s_{0}}$ holds. Else, the set $A_{s_{0}}$ is countable, pick $m_{0}\in \mathbb{W}$ such that $(m_{0}, s_{0})\in A$. Let $M = \downarrow\!(A\setminus \downarrow\!(m_{0}, s_{0}))$. Then $A\subseteq\ \downarrow\!(m_{0}, s_{0})\bigcup M$. Now we construct a Scott closed set $W$ containing $M$. For $M$, we know $\max M$ exists and $\max M = \max (A\setminus \downarrow\!(m_{0}, s_{0}))$. Let
\begin{center}
$E = \{s\in f(\max M): M_{s}\ \mathrm{is}\ \mathrm{countable}\}$, \quad $B = \{s\in E: \sup M_{s}\notin M_{s}\}$,
\end{center}
where $M_{s} = \{m\in \mathbb{W}: (m, s)\in M\}$.
And for any $s\in B$, we set
\begin{center}
$T_{s} = \{(\vee M_{s}).t.p: p = s\ \mathrm{or}\ p<_{1} s,\ \vee M_{s}\leq\min t\ \mathrm{or}\ t = \varepsilon\}$.
\end{center}
Now let $W = \bigcup_{s\in f(\max M)}W_{s}$, where $\{W_{s}: s\in f(\max M)\}$ are defined as follows:
\begin{itemize}
\item $W_{s} = \downarrow\!L_{s}$ if $s\in f(\max M)\setminus E$
\item $W_{s} = \bigcup_{a\in M_{s}}\downarrow\!(a, s)$ if $s\in E\setminus B$
\item $W_{s} = \bigcup_{a\in M_{s}}\downarrow\!(a, s)\bigcup\,(\bigcup\{\downarrow\!L_{\vee C}: C\subseteq \bigcup_{s\in B}T_{s}\ \mathrm{is\ a\ chain}\})$ if $s\in B$.
\end{itemize}

$\mathbf{Claim ~1}$: If $C\subseteq \bigcup_{s\in B}T_{s}$ is a non-trivial chain, then $C$ is countable.

Assume that there exists an uncountable non-trivial chain $C = \{u.s: u\in D\}$ contained in $\bigcup_{s\in B}T_{s}$. Then for any $u.s\in C$, there is a $s_{u}\in B$ such that $u.s\in T_{s_{u}}$, that is, $u.s = (\vee M_{s_{u}}).t_{u}.p_{u}$ with $p_{u} = s_{u}\ \mathrm{or}\ p_{u}\sqsubset_{1} s_{u},\ \vee M_{s_{u}}\leq\min t_{u}\ \mathrm{or}\ t_{u} = \varepsilon$. 
It is easy to verify that $s_{u'}\neq s_{u''}$ and $M_{s_{u'}}\neq M_{s_{u''}}$if $u'\neq u''$ for any $u',u''\in D$.
In addition, note that $t_{u}.p_{u} = s$ for each $u\in D$, which means that $s\sqsubseteq p_u$. The fact that $p_u\sqsubseteq s_u$ for any $u\in D$ suggests that $\{s_{u}: u\in D\}\subseteq\ \uparrow\!s$. This reveals that $\{s_u\mid u\in D\}$ is a chain by Proposition \ref{sf}, which ensures that we can find a cofinal subset $\{s_{u}: u\in D'\}$ of $\{s_{u}: u\in D\}$ in which there is only the relation $\sqsubset_{1}$ among all elements. 
For any $u_1,u_2\in D'$ with $u_1< u_2$, it turns out that $\vee M_{s_{u_{1}}}< \vee M_{s_{u_{2}}}$. Note that for any $s_{1}, s_{2}\in \mathbb{W}^*$ with $s_{1}\sqsubseteq  s_{2}$, we can conclude that $M_{s_2}\subseteq M_{s_1}$ because $(m,s_1)\leq (m,s_2)$ for any $m\in \mathbb{W}$ and $M$ is a lower set. Then the fact that $\{s_{u}: u\in D'\}$ is a chain infers that $M_{s_{u'}}\subsetneqq M_{s_{u''}}$ or $M_{s_{u''}}\subsetneqq M_{s_{u'}}$ for any $u',u''\in D'$ with $u'\neq u''$. Due to the assumption that $\vee M_{s_{u_{1}}}< \vee M_{s_{u_{2}}}$, we can deduce that $M_{s_{u_{1}}}\subsetneqq M_{s_{u_{2}}}$, so $s_{u_{1}}\!\sqsupseteq\! s_{u_{2}}$. Now for a fixed $u_{0}\in D'$, we have $\{s_{u}: u\in D', u> u_{0}\}\subseteq \{s\in \mathbb{W}^{*}: s\sqsubset_{1} s_{u_{0}}\}$. Because $D$ is uncountable and $u = \vee M_{s_{u}}$ for any $u\in D$, the set $\{s_{u}: u\in D\}$ is uncountable. It follows that $\{s_{u}: u\in D'\}$ is uncountable. Furthermore, $\{s_{u}: u\in D', u> u_{0}\}$ is uncountable, which contradicts the fact that $\{s\in \mathbb{W}^{*}: s\sqsubset_{1} s_{u_{0}}\}$ is countable.

$\mathbf{Claim~2}$: $W = \bigcup_{s\in f(\max M)}W_{s}$ is Scott closed.

One sees clearly that $W$ is a lower set. So it remains to confirm that $W$ is closed for the $\mathrm{sups}$ of any non-trivial chain.

Let $\{(m, u.s): u\in D_{0}\}$ be a non-trivial chain in $W$, where $D_{0}\subseteq \mathbb{W}$. Now we need to distinguish the following cases.

Case 1, $\{(m, u.s): u\in D_{0}\}\bigcap\ (\bigcup_{s\in E}\bigcup_{a\in M_{s}}\downarrow\!(a, s))$ is cofinal in $\{(m, u.s): u\in D_{0}\}$. For any $u\in D_{0}$, there are $s_{u}\in E$ and $a_{u}\in M_{s_{u}}$ such that $(m, u.s)< (a_{u}, s_{u})$, without loss of generality, assume that $(a_{u}, s_{u})\in \max M$. We set
\begin{center}
$A_{<_{r}} = \{(m, u.s): u\in D_{0}, \exists (a_{u}, s_{u})\in \max M\ \mathrm{with}\ s_{u}\in E\ s.t.\ (m, u.s)<_{r} (a_{u}, s_{u})\}$
\end{center}
and
\begin{center}
$U_{<_{r}} = \{u\in D_{0}: (m, u.s)\in A_{<_{r}}\}$,
\end{center}
where $<_{r}\in \{<_{1},\ <_{2},\ <_{3},\ <_{2};<_{1},\ <_{3};<_{1}\}$.

Case 1.1, $A_{<_{1}}$ is cofinal in $\{(m, u.s): u\in D_{0}\}$. For any $(m, u.s)\in A_{<_{1}}$, there exists $(a_{u}, s_{u})\in \max M$ with $s_{u}\in E$ such that $(m, u.s)<_{1} (a_{u}, s_{u})$. Then $a_{u} = m$ and there is a unique $(m, s_{u_{0}})$ such that $(m, u.s)<_{1} (m, s_{u_{0}})$ since $(a_{u}, s_{u})\in \max M$ for each $u\in U_{<_{1}}$. Thus $\sup_{u\in D_{0}}(m, u.s)\leq (m, s_{u_{0}})\in W$.

Case 1.2, $A_{<_{2}}$ or $A_{<_{2};<_{1}}$ is cofinal in $\{(m, u.s): u\in D_{0}\}$. It can be easily verified that $\sup_{u\in D_{0}}(a, u.s)\in W$ in both cases, so we omit the proof here.

Case 1.3, $A_{<_{3}}$ is cofinal in $\{(m, u.s): u\in D_{0}\}$.  For any $(m, u.s)\in A_{<_{3}}$, there exists $(a_{u}, s_{u})\in \max M$ with $s_{u}\in E$ such that $(m, u.s)<_{3} (a_{u}, s_{u})$. We can find $s_{u_{0}}$ such that $\{(m, u.s): u\in U_{<_{3}}, (m, u.s)<_{3} (a_{u}, s_{u_{0}})\}$ is cofinal in $A_{<_{3}}$ by reason that $s_{u}\subseteq s$ for each $u\in U_{<_{3}}$ and the length of $s$ is finite. If $s_{u_{0}}\in E\setminus B$, that is, $\vee M_{s_{u_{0}}}\!\in M_{s_{u_{0}}}$, then $a_u\leq \vee M_{s_{u_{0}}}$. It follows that $(m, u.s)<_{3} (\vee M_{s_{u_{0}}}, s_{u_{0}})$ for each $u$, which implies that $\sup_{u\in D_{0}}(m, u.s)\leq (\vee M_{s_{u_{0}}}, s_{u_{0}})\in M$. Hence, $\sup_{u\in D_{0}}(m, u.s)\in W$. Suppose that $s_{u_{0}}\in B$, that is, $\vee M_{s_{u_{0}}}\!\!\notin M_{s_{u_{0}}}$. Under this condition, we assume $s = ts_{u_{0}}$,
where $t = \varepsilon$ or $t\neq \varepsilon$. In the case that $t\neq \varepsilon$ and there exists a $u_{1}> \min t$, then $\min t\leq a_{u_{1}}$, and so $\sup_{u\in D_{0}}(m, u.s)<_{3} (a_{u_{1}}, s_{u_{0}})$. In the case that $t = \varepsilon$ or $u\leq \min t$ for all $u$, we can get that $u\leq a_{u}$ for each $u$. Since $a_{u}\in M_{s_{u_{0}}}$ and $M_{s_{u_{0}}}$ is countable, $D_{0}$ is countable under each case and $\sup U_{<_{3}}\leq \vee M_{s_{u_{0}}}$. Thus we have $\sup_{u\in D_{0}}(m, u.s) = (m, \vee U_{<_{3}}.t.s_{u_{0}})\leq (m, \vee M_{s_{u_{0}}}.t.s_{u_{0}})$. Note that $\vee M_{s_{u_{0}}}.t.s_{u_{0}}\in T_{s_{u_{0}}}$. Hence, $\sup_{u\in D_{0}}(m, u.s)\in W$.

Case 1.4, $A_{<_{3};<_{1}}$ is cofinal in $\{(m, u.s): u\in D_{0}\}$. The argument of this case is similar to the above case.

Case 2, $\{(m, u.s): u\in D_{0}\}\bigcap\ (\bigcup_{s\in f(\max M)\setminus E}\downarrow\!L_{s})$ is cofinal in $\{(m, u.s): u\in D_{0}\}$. For any $u\in D_{0}$, there are $s_{u}\in f(\max M)\setminus E, a_{u}\in \mathbb{W}$ such that $(m, u.s)< (a_{u}, s_{u})$. We set
\begin{center}
$A_{<_{r}} = \{(m, u.s): u\in D_{0}, \exists s_{u}\in f(\max M)\setminus E, a_{u}\in \mathbb{W}\ s.t.\ (m, u.s)<_{r} (a_{u}, s_{u})\}$
\end{center}
and
\begin{center}
$U_{<_{r}} = \{u\in D_{0}: (m, u.s)\in A_{<_{r}}\}$,
\end{center}
where $<_{r}\in \{<_{1},\ <_{2},\ <_{3},\ <_{2};<_{1},\ <_{3};<_{1}\}$.

Case 2.1, $A_{<_{1}}$ is cofinal in $\{(m, u.s): u\in D_{0}\}$. For any $(m, u.s)\in A_{<_{1}}$, there exist $s_{u}\in f(\max M)\setminus E $ and $a_{u}\in \mathbb{W}$ such that $(m, u.s)<_{1} (a_{u}, s_{u})$. Then $a_{u} = m, u.s\sqsubset_{1} s_{u}$. It turns out that $s_{u'}=s_{u''}$ or $s_{u'}\sqsubset_{1}s_{u''}$. As $s_{u}\in f(\max M)\setminus E$, that is, $M_{s_{u}}$ is uncountable for any $u\in D_{0}$. Next, we claim that $M_{s_{u}} = \mathbb{W}\setminus \{m_{0}\}$. Assume that there is a $m\in M_{s_{u}}$ such that $m = m_{0}$. By $(m, s_{u})\in M = \downarrow\!(A\setminus \downarrow\!(m_{0}, s_{0}))$, we have $(m, s_{u})\leq (m_{1}, s_{1})$ for some $(m_{1}, s_{1})\in A\setminus \downarrow\!(m_{0}, s_{0})$. This indicates that $m_{1}\neq m_{0}$, or else, $(m_{1}, s_{1})\leq (m_{0}, s_{0})$ since $s_{1}\in f(A)$ and $s_{0} = \sup f(A)$, a contradiction. Thus $(m, s_{u})<_{3} (m_{1}, s_{1})$ or $(m, s_{u})<_{3};<_{1} (m_{1}, s_{1})$. The assumption that $s_{u}\in f(\max M)$ guarantees the existence of $n_{u}\in \mathbb{W}$ such that $(n_{u}, s_{u})\in \max M$. It follows that $(n_{u}, s_{u})<_{3} (m_{1}, s_{1})$ or $(n_{u}, s_{u})<_{3};<_{1} (m_{1}, s_{1})$, which will result in a contradiction since there would be a maximal element of $M$ being strictly less than an element of $M$. Therefore, $M_{s_{u}}\subseteq \mathbb{W}\setminus \{m_{0}\}$. For the converse, By reason that $A_{s_{u}}$ contains  a uncountable set $M_{s_{u}}$, we know that $L_{s_{u}}\subseteq A$. 
Let $m\in  \mathbb{W}\setminus \{m_{0}\}$. If $(m, s_{u})\leq (m_{0}, s_{0})$, then $(m, s_{u})<_{3} (m_{0}, s_{0})$ or $(m, s_{u})<_{3};<_{1} (m_{0}, s_{0})$ since $m \neq m_{0}$. By the assumption that $s_{u}\in f(\max M)$ again, we can find $c_{u}\in \mathbb{W}$ satisfying $(c_{u}, s_{u})\in \max M$. Then $(c_{u}, s_{u})<_{3} (m_{0}, s_{0})$ or $(c_{u}, s_{u})<_{3};<_{1} (m_{0}, s_{0})$, but this will contradict to the fact that $(c_{u}, s_{u})\in \max M\subseteq A\setminus \downarrow\!(m_{0}, s_{0})$. Hence, we can obtain $\mathbb{W}\setminus \{m_{0}\} = M_{s_{u}}$. Now suppose that there are $u_{1}, u_{2}\in D_{0}$ such that $s_{u_{1}}\sqsubset_{1}s_{u_{2}}$. 
From $s_{u_{1}}\in f(\max M)$, we can identify $(n_{1}, s_{u_{1}})\in \max M$ for some $n_{1}\in \mathbb{W}\setminus \{m_{0}\}$. Furthermore, $n_{1}\in M_{s_{u_{2}}}$ as $M_{s_{u}} = \mathbb{W}\setminus \{m_{0}\}$ for any $u\in D_{0}$, this is to say $(n_{1}, s_{u_{2}})\in M$. Note that we have assumed that $s_{u_{1}}\sqsubset_{1}s_{u_{2}}$, we have $(n_{1}, s_{u_{1}})<_{1} (n_{1}, s_{u_{2}})$, a contradiction. So for any $u\in D_{0}$, there exists a unique $s^{*}\in f(\max M)\setminus E$ such that $(m, u.s)<_{1} (m, s^{*})$, which leads to the result that $\sup_{u\in D_{0}}(m, u.s)\leq (m, s^{*})$, and we can conclude that $\sup_{u\in D_{0}}(m, u.s)\in W$.

Case 2.2, $A_{<_{2}}$ or $A_{<_{2};<_{1}}$ is cofinal in $\{(m, u.s): u\in D_{0}\}$. One sees immediately  that $\sup_{u\in D_{0}}(m, u.s)\in W$ in the case.

Case 2.3, $A_{<_{3}}$ is cofinal in $\{(m, u.s): u\in D_{0}\}$. For any $(m, u.s)\in A_{<_{3}}$, there exist $s_{u}\in f(\max M)\setminus E$ and $\ a_{u}\in \mathbb{W}$ such that $(m, u.s)<_{3} (a_{u}, s_{u})$. Then we have $s_{u}\subseteq s$ for any $u\in U_{<_{3}}$, which implies $s\sqsubseteq s_{u}$. Thus for any given $u'\in U_{<_{3}}$, $\sup_{u\in D_{0}}(m, u.s)=\sup_{u\in U_{<_{3}}}(m, u.s)\leq (m, s_{u'})\in L_{s_{u'}}\subseteq W$.

Case 2.4, $A_{<_{3};<_{1}}$ is cofinal in $\{(m, u.s): u\in D_{0}\}$. The proof is similar to Case $2.3$.

Case 3, $\{(m, u.s): u\in D_{0}\}\bigcap\ (\bigcup \{\downarrow\!L_{\vee C}: C\subseteq \bigcup_{s\in B}T_{s}\ \mathrm{is}\ \mathrm{a}\ \mathrm{chain}\})$ is cofinal in $\{(m, u.s): u\in D_{0}\}$. For any $u\in D_{0}$, there exist $C_{u}\subseteq \bigcup_{s\in B}T_{s}\ \mathrm{and}\ m_{u}\in \mathbb{W}$ such that $(m, u.s)< (m_{u}, \vee C_{u})$. Set
\begin{center}
$A_{<_{r}} = \{(m, u.s): u\in D_{0}, \exists\ C_{u}\subseteq \bigcup_{s\in B}T_{s}\ \mathrm{and}\ m_{u}\in \mathbb{W}\ s.t.\ (m, u.s)<_{r} (m_{u}, \vee C_{u})\}$
\end{center}
and
\begin{center}
$U_{<_{r}} = \{u\in D_{0}: (m, u.s)\in A_{<_{r}}\}$,
\end{center}
where $<_{r}\in \{<_{1},\ <_{2},\ <_{3},\ <_{2};<_{1},\ <_{3};<_{1}\}$.

Case 3.1, $A_{<_{1}}$ is cofinal in $\{(m, u.s): u\in D_{0}\}$. For any $(m, u.s)\in A_{<_{1}}$, there are $C_{u}\subseteq \bigcup_{s\in B}T_{s}\ \mathrm{and}\ m_{u}\in \mathbb{W}$ such that $(m, u.s)<_{1} (m_{u}, \vee C_{u})$. Then $m_{u} = m$ and $u.s<_{1} \vee C_{u}$. We know each $C_{u}$ is countable by Claim 1, so there exists a chain $E_u$ of the form $\{a.s: a\in D_{u}\}$ is cofinal in $C_u$. It follows that $\bigcup_{u\in U_{<_{1}}}E_u$ is a chain in $\bigcup_{s\in B}T_{s}$. Again applying Claim 1, we know that $\bigcup_{u\in U_{<_{1}}}E_u$ is a countable chain. So there exists a chain $C'$ of the form $\{a.s: a\in D'\}$ contained in $\bigcup_{u\in U_{<_{1}}}C_{u}$ as a cofinal subset, and $C'\subseteq \bigcup_{s\in B} T_{s}$. Hence, we can get that $\sup_{u\in D'}(m, u.s)\leq (m, \vee C')$, and hence $\sup_{u\in D_{0}}(m, u.s)\in W$.

The proof of the residual cases are similar to that of the case 2.2, case 2.3, case 2.4, respectively.
Therefore, $W$ is Scott closed. Because $A\subseteq \downarrow\!(m_{0}, s_{0})\bigcup M$ and $M\subseteq W$, $A\subseteq\ \downarrow\!(m_{0}, s_{0})\bigcup W$. So our proof will be complete if we can illustrate that $A$ cannot be contained in $W$. Suppose not, $A\subseteq W$, then $(m_{0}, s_{0})\in W$. 

\begin{itemize}
\item  If $(m_{0}, s_{0})\in \bigcup_{s\in E}\bigcup_{a\in M_{s}}\downarrow\!(a, s)$, then there is $(a_{1}, s_{1})\in \max M$ with $s_{1}\in E$ such that $(m_{0}, s_{0})\leq (a_{1}, s_{1})$, so $s_{0}\sqsubseteq s_{1}$. Since $s_{1}\in f(\max M)\subseteq f(A)$, $s_{1}\sqsubseteq s_{0}$, that is, $s_{1} = s_{0}$. Thus $m_{0} = a_{1}$ and so $(m_{0}, s_{0}) = (a_{1}, s_{1})$, but this contradicts to the fact that $(a_{1}, s_{1})\in \max M\subseteq A\setminus \downarrow\!(m_{0}, s_{0})$. 
\item  If $(m_{0}, s_{0})\in \bigcup_{s\in f(\max M)\setminus E}\downarrow\!L_{s}$, then there are $m'\in \mathbb{W},\ s'\in f(\max M)\setminus E$ such that $(m_{0}, s_{0})\leq (m', s')$. This yields that $s_0\leq s'$. As $M_{s'}$ is uncountable, so is $A_{s'}$. This means that $s'\in f(A)$, which induces $s'\sqsubseteq s_{0}$. As a result, $s' = s_{0}$. It follows that $s_{0}\in f(\max M)\setminus E$, so we know that $M_{s_{0}}$ is uncountable, which indicates that $A_{s_{0}}$ is uncountable, but this violates the assumption that $A_{s_{0}}$ is countable. 
\item If $(m_{0}, s_{0})\in \bigcup\{\downarrow\!L_{\vee C}: C\subseteq \bigcup_{s\in B}T_{s}\ \mathrm{is\ a\ chain}\}$, then there are $C'\subseteq \bigcup_{s\in B}T_{s}$ being a countable chain and $m'\in \mathbb{W}$ with $(m_{0}, s_{0})\leq (m', \vee C')$. For any $c\in C'$, there is a $s_{c}\in B$ such that $c = \vee M_{s_{c}}.t_{c}.p_{c}$, where $\vee M_{s_{c}}\leq \min t_{c}$, or $t=\varepsilon$ and $\ p_{c} = s_{c}$ or $p_{c} \sqsubset_{1} s_{c}$. Then there must exist $C^{*}$ being a cofinal subset of $C'$, where $c = \vee M_{s_{c}}.t_{c}.p_{c} = \vee M_{s_{c}}.s^{*}$ for any $c\in C^{*}$. Because for any $c\in C^{*}$, we have $s_{c}\in B\subseteq f(\max M)\subseteq f(A)$, $s_{c}\sqsubseteq s_{0}$. This implies that $c\sqsubseteq s_c\sqsubseteq s_{0}$ for each $c\in C^{*}$, thus $\vee C^{*}\sqsubseteq s_{0}$. As $(m_{0}, s_{0})\leq (m', \vee C')$, $s_{0}\sqsubseteq \vee C' = \vee C^{*}$. We can conclude that $s_{0} = \vee C^{*} = \bigvee_{c\in C^{*}}\vee M_{s_{c}}.s^{*}$. This means that $|p_{c}|\leq |s^{*}|< |s_{0}|$, where  $|s|$ denotes the string length of $s$ for any $s\in \mathbb{W}^*$. But from $p_{c}\sqsubseteq s_{c}\sqsubseteq s_{0}$, we know $|p_{c}|\geq |s_{0}|$, a contradiction. Hence, $A\not\subseteq W$, and so $A\subseteq\ \downarrow\!(m_{0}, s_{0})$.
\end{itemize}

In a conclusion, for every irreducible closed subset $A$ of $\mathcal{Z}$,  we have proved that either $A = {\downarrow}(m_{0}, s_{0})$ or $A = \downarrow\!L_{s_{0}}$.
\end{proof}
\end{theorem}

\subsection{$\mathcal{Z}$ is well-filtered}
In this subsection, we confirm that $\mathcal{Z}$ is well-filtered. First, we consider the compact saturated subsets of $\mathcal{Z}$. 

\begin{lemma}\label{min1}
If $K$ is a compact saturated subset of $\mathcal{T}$, then $\min K$ is finite, that is, $K = \uparrow\!F$ for some finite subset $F\subseteq\! \mathbb{W}^{*}$.
\begin{proof}
Assume for the sake of a contradiction that $\min K$ is infinite. Then for any finite subset $F$ contained in $\min K$, we set $B_{F} =\ \downarrow\!(\min K\setminus F)$. We claim that $B_{F}$ is Scott closed. To this end, let $\{u.s: u\in D\}$ be a non-trivial chain contained in $B_{F}$. For any $u\in D$, there is a $s_{u}\in \min K\setminus F$ such that $u.s\sqsubset s_{u}$. Assume that $u_{1}< u_{2}$ for some $u_{1}, u_{2}\in D$, then $\{s_{u_{1}}, s_{u_{2}}\}\subseteq\ \uparrow\!\!u_{1}.s$. So $s_{u_{1}}\sqsubseteq s_{u_{2}}$ or $s_{u_{2}}\sqsubseteq s_{u_{1}}$ from Proposition \ref{sf}. Since $\{s_{u_{1}}, s_{u_{2}}\}\subseteq\min K$, $s_{u_{1}} = s_{u_{2}}$. Thus for any $u\in D$, there exists a unique $s_{0}\in \min K\setminus F$ such that $u.s\sqsubset s_{0}$. It follows that $\sup_{u\in D}u.s\sqsubseteq s_{0}$, which yields that $\sup_{u\in D}u.s\in B_{F}$. The set $B_F$ is obviously a lower set, and therefore, $B_{F}$ is Scott closed. Note that $\{B_{F}: F\subseteq_{f} \min K\}$ is filtered and $B_{F}\cap \min K\neq \emptyset$, where $F\subseteq_f\min K$ denotes $F$ is a finite subset of $\min K$. Then we arrive at $ \bigcap_{F\subseteq_{f} \min K}B_{F} \bigcap \min K\neq \emptyset$. Pick $s\in \bigcap_{F\subseteq_{f} \min K}B_{F} \bigcap \min K$, which means that $s\in B_{\{s\}} = \da (\min K\setminus \{s\})$, which is absurd.
\end{proof}
\end{lemma}

\begin{lemma}\label{min2}
If $K$ is a compact saturated subset of $\mathcal{Z}$, then for any $s\in f(\min K)$, $K_{s}$ is countable, where $K_{s} = \{m\in \mathbb{W}: (m, s)\in \min K\}$.
\begin{proof}
Assume that there is a $s_{0}\in f(\min K)$ such that $K_{s_{0}}$ is uncountable. Then pick $M = \{m_{n}: n\in \mathbb{N}\}\subseteq K_{s_{0}}$, where $\mathbb{N}$ denotes the set of all positive integers. We know that $M$ is countable, which ensures the existence of $\vee M$. Now for each $n\in \mathbb{N}$, we set
\begin{center}
$B_{n} = \bigcup_{i\in \mathbb{N}\setminus \{1,2,...,n-1\}}\downarrow\!(m_{i}, s_{0})\ \bigcup\ (\bigcup\{\downarrow\!L_{s}: s\in T\})$,
\end{center}
where $T = \{\vee M.t.a: t = \varepsilon\ \mathrm{or}\vee M\leq\min t, a = s_{0}\ \mathrm{or}\ a\sqsubset_{1} s_{0}\}$.

$\mathbf{Claim}$: For each $n\in \mathbb{N}$, $B_{n}$ is Scott closed.

Let $\{(a, u.s): u\in D\}$ be a non-trivial chain contained in $B_{n}$.

Case 1, $\{(a, u.s): u\in D\}\bigcap\, (\bigcup_{i\in \mathbb{N}\setminus \{1,2,...,n-1\}}\downarrow\!(m_{i}, s_{0}))$ is cofinal in $\{(a, u.s): u\in D\}$. Then for any $u\in D$, there is $m_{u}\in M\setminus \{m_{1},m_{2},m_{3},..., m_{n-1}\}$ such that $(a, u.s)< (m_{u}, s_{0})$. For convenience we set
\begin{center}
$A_{<_{r}} = \{(a, u.s): u\in D, \exists m_{u}\in M\setminus \{m_{1},m_{2},m_{3},..., m_{n-1}\}\ s.t.\ (a, u.s)<_{r} (m_{u}, s_{0})\}$,
\end{center}
where $<_{r}\in \{<_{1},\ <_{2},\ <_{3},\ <_{2};<_{1},\ <_{3};<_{1}\}$.

Case 1.1, $A_{<_{1}}$ is cofinal in $\{(a, u.s): u\in D\}$. For any $(a, u.s)\in A_{<_{1}}$, $(a, u.s)<_{1} (m_{u}, s_{0})$ for some $m_{u}\in M\setminus \{m_{1},m_{2},m_{3},..., m_{n-1}\}$, which implies that $m_{u} = a$, and $(a, u.s)<_{1} (a, s_{0})$. So $\sup_{u\in D}(a, u.s)\leq (a, s_{0})\in B_n$.

Case 1.2, $A_{<_{2}}$ or $A_{<_{2};<_{1}}$ is cofinal in $\{(a, u.s): u\in D\}$. In this case, for any $u\in D$, $(a, u.s)$ is less than a fixed point $(m_{0}, s_{0})$, where $a = m_{0}\in M\setminus \{m_{1},m_{2},m_{3},..., m_{n-1}\}$. Hence $\sup_{u\in D}(a, u.s)\in B_{n}$.

Case 1.3, $A_{<_{3}}$ is cofinal in $\{(a, u.s): u\in D\}$. For any $(a, u.s)\in A_{<_{3}}$, $(a, u.s)<_{3} (m_{u}, s_{0})$ for some $m_{u}\in M\setminus \{m_{1},m_{2},m_{3},..., m_{n-1}\}$. Then $s = ts_{0}$ for some $t\in \mathbb{W}^{*}$. If $t\neq \varepsilon$ and there is a $u_{0}> \min t$, then $\min t\leq m_{u_{0}}$ and we can get that $\sup_{u\in D}(a, u.s)<_{3} (m_{u_{0}}, s_{0})$. If $t = \varepsilon$ or $u\leq\min t$ for all $u$, then $D$ is countable and $u\leq m_{u}$.
So we have $\sup D\leq \vee M$ and $\sup D\leq \min t$ or $t=\varepsilon$.  Now we need to distinguish the following two cases.

Case 1.3.1, $\sup D=\vee M$. Then $\vee M\leq \min t$ or $t=\varepsilon$. This implies that $\vee M.t.s_{0}\in T$. One sees immediately that $\sup_{u\in D}(a, u.s) \leq (a, \vee M.t.s_{0})\in B_{n}$. Hence, $\sup_{u\in D}(a, u.s)\in B_{n}$.

Case 1.3.2, $\sup D<\vee M$. Then we can find $m_{n_{1}}\in M\setminus \{m_{1},m_{2},m_{3},..., m_{n-1}\}$ such that $\sup D\leq m_{n_{1}}$. This means that $\sup_{u\in D}(a, u.s)<_{3}(m_{n_{1}}, s_{0})$, and we also have that $\sup_{u\in D}(a, u.s)\in B_n$.

Case 1.4, $A_{<_{3};<_{1}}$ is cofinal in $\{(a, u.s): u\in D\}$. For any $(a, u.s)\in A_{<_{3};<_{1}}$, $(a, u.s)<_{3};<_{1} (m_{u}, s_{0})$ for some $m_{u}\in M\setminus \{m_{1},m_{2},m_{3},..., m_{n-1}\}$. That is, there is an $s_{u}\in \mathbb{W}^{*}$, such that $(a, u.s)<_{3} (m_{u}, s_{u})<_{1} (m_{u}, s_{0})$. Then there is a fixed $s^{*}\in \mathbb{W}^{*}$ such that $\{(a, u.s): \exists\ m_{u}\in M\setminus \{m_{1},m_{2},m_{3},..., m_{n-1}\}\ s.t.\ (a, u.s)<_{3} (m_{u}, s^{*})<_{1} (m_{u}, s_{0})\}$ being a cofinal subset of $A_{<_{3};<_{1}}$ since each $s_{u}\subseteq s$ and the length of $s$ is finite. Then $s = ts^{*}$ and the residual analysis is similar to that of the above case.

Case 2, $\{(a, u.s): u\in D\}\bigcap\ (\bigcup\{\downarrow\!L_{s}: s\in T\})$ is cofinal in $\{(a, u.s): u\in D\}$. That is, for any $u\in D$, there exists $m_{u}\in \mathbb{W}$ and $s_{u} = \vee M.t_{u}.p_{u}\in T$ such that $(a, u.s)< (m_{u}, s_{u})$, with $t = \varepsilon\ \mathrm{or}\vee M\leq\min t, p_u = s_{0}\ \mathrm{or}\ p_u\sqsubset_{1} s_{0}$. Similar to the above case, we set
\begin{center}
$A'_{<_{r}} = \{(a, u.s): u\in D, \exists m_{u}\in\mathbb{W}, s_{u} = \vee M.t_{u}.p_{u}\in T\ s.t.\ (a, u.s)<_{r} (m_{u}, s_{u})\}$,
\end{center}
where $<_{r}\in \{<_{1},\ <_{2},\ <_{3},\ <_{2};<_{1},\ <_{3};<_{1}\}$.

Case 2.1, $A'_{<_{1}}$ is cofinal in $\{(a, u.s): u\in D\}$. For any $(a, u.s)\in A'_{<_{1}}$, there exist $m_{u}\in \mathbb{W}$ and $s_{u} = \vee M.t_{u}.p_{u}\in T$ such that $(a, u.s)<_{1} (m_{u}, s_{u})$. Then $m_{u} = a,\ u.s\sqsubset_{1} s_{u} = \vee M.t_{u}.p_{u}$. Thus $t_{u}.p_{u} = s$ for each $u$. It follows that $(m_{u}, s_{u})=(a, \vee M.s)$ for any $u\in D$. This means that $\sup_{u\in D}(a, u.s)\leq (a, \vee M.s)$.  Hence, $\sup_{u\in D}(a, u.s)\in B_n$.

Case 2.2, $A'_{<_{2}}$ or $A'_{<_{2};<_{1}}$ is cofinal in $\{(a, u.s): u\in D\}$. For the case, it is easy to verify that $\sup_{u\in D}(a, u.s)\in B_{n}$.

Case 2.3, $A'_{<_{3}}$ is cofinal in $\{(a, u.s): u\in D\}$. For any $(a, u.s)\in A'_{<_{3}}$, there exist $m_{u}\in \mathbb{W}$ and $s_{u} = \vee M.t_{u}.p_{u}\in T$ such that $(a, u.s)<_{3} (m_{u}, s_{u})$. Then we have $s_{u}\subseteq s$, that is $s\sqsubseteq s_{u}$. Thus for any fixed $(a, u_0.s)\in A'_{<_{3}}$, we have $\sup_{u\in D}(a, u.s)\leq (a, s_{u_0})\in B_n$. Therefore, $\sup_{u\in D}(a, u.s)\in B_n$.

Case 2.4, $A'_{<_{3};<_{1}}$ is cofinal in $\{(a, u.s): u\in D\}$. The proof is similar to Case $2.3$.

One sees immediately that $B_n$ is a lower set for any $n\in \mathbb{N}$. Thus $B_{n}$ is Scott closed for each $n\in \mathbb{N}$. 

Since $\{B_{n}: n\in \mathbb{N}\}$ is filtered and $B_{n}\cap \min K\neq \emptyset$, $\bigcap_{n\in \mathbb{N}}B_{n}\bigcap \min K\neq\emptyset$. Choose $(m, s)\in \bigcap_{n\in \mathbb{N}}B_{n}\bigcap \min K$. We claim that $(m, s)\notin \downarrow\!L_{t}$ for any $t\in T$. Suppose not, there are $m'\in \mathbb{W}, s'\in T$ such that $(m, s)\leq (m', s')$. Assume that $s' = \vee M.t'.p'$ with $\vee M\leq \min t' $ or $t=\varepsilon$, and $p' = s_{0}$ or $p'\sqsubset_{1} s_{0}$, then we can obtain that $\ua \vee M\cap K_{s_0} \neq \emptyset$ from the uncountability of $K_{s_0}$. It follows that 
\begin{center}
$(m, s)\leq (m', s') = (m', \vee M.t'.p')<_{3} (m_{1}, p')\leq (m_{1}, s_{0})$
\end{center}
for some $m_{1}\in \ua \vee M\cap K_{s_{0}}$, which implies contradiction as $(m_{1}, s_{0})\in \min K$. Now we suppose that $(m, s)\in B_{n_{0}}$ for some $n_{0}\in \mathbb{N}$, then $(m, s)\in \bigcup_{i\in \mathbb{N}\setminus \{1,2,...,n_{0}-1\}}\downarrow\!(m_{i}, s_{0})$. This means that there is an $i_{0}\geq n_{0}$ such that $(m, s)\leq (m_{i_{0}}, s_{0})$, and hence, $(m, s) = (m_{i_{0}}, s_{0})$ since both of them are minimal elements of $K$. The fact that $(m, s)\in \bigcap_{n\in \mathbb{N}}B_{n}$ indicates that $(m, s)\in B_{i_{0}+1}$, that is, $(m, s)\in \bigcup_{i\in \mathbb{N}\setminus \{1,2,...,i_{0}\}}\downarrow\!(m_{i}, s_{0})$. However, this is impossible. Therefore, the original hypothesis doesn't hold. So $K_{s}$ is countable for any $s\in f(\min K)$.
\end{proof}
\end{lemma}

\begin{lemma}\label{min3}
If $K$ is a compact saturated subset of $\mathcal{Z}$, then for any $s\in \mathbb{W}^{*}$, $(\bigcup_{u\in \mathbb{W}}L_{u.s})\bigcap \min K$ is countable, where $L_{u.s} = \{(m, u.s): m\in \mathbb{W}\}$.
\begin{proof}
Suppose that there exists $s_{0}\in \mathbb{W}^{*}$ such that $(\bigcup_{u\in \mathbb{W}}L_{u.s_{0}})\bigcap \min K$ is uncountable. Let
\begin{center}
$A = \{u.s_{0}: \exists\ m_{u}\in \mathbb{W}\ s.t.\ (m_{u}, u.s_{0})\in \min K\}$
\end{center}
By Lemma \ref{min2}, we have $L_{u.s_0}\cap \min K$ is countable for any $u\in \mathbb{W}$, which implies that $A$ must be uncountable. Then pick $\{u_{n}.s_{0}: n\in \mathbb{N}\}$ as a countable subset of $A$, besides, for each $n\in \mathbb{N}$, we extract $m_{n}\in \mathbb{W}$ so that $(m_{n}, u_{n}.s_{0})\in \min K$, where $\mathbb{N}$ denotes all positive integers. Then set
\begin{center}
$G = \bigcup_{n\in \mathbb{N}}\downarrow\!(m_{n}, u_{n}.s_{0})$
\end{center}
and
\begin{center}
$E_{n} = G_{u_{n}.s_{0}} = \{a\in \mathbb{W}: (a, u_{n}.s_{0})\in G\}$.
\end{center}
Note that each $E_{n}$ is countable and $E_{n_{1}}\subseteq E_{n_{2}}$ if $n_{1}\geq n_{2}$. Now we construct $B_{n}\ (\forall n\in \mathbb{N})$ as follows:
\begin{center}
$B_{n} = \bigcup_{i\in \mathbb{N}\setminus \{1,2,...,n-1\}}\bigcup_{a\in E_{i}}\downarrow\!(a, u_{i}.s_{0})\bigcup\ (\bigcup \{\downarrow L_{\vee C}: C\subseteq \bigcup_{n\in \mathbb{N}} T_{u_{n}.s_{0}}\ \mathrm{is\ a\ chain}\})$,
\end{center}
where $T_{u_{n}.s_{0}} = \{\vee E_{n}.t.p: t = \varepsilon\ \mathrm{or}\ \vee E_{n}\leq\min t,\ p = u_{n}.s_{0}\ \mathrm{or}\  p\sqsubset_{1} u_{n}.s_{0}\}$.

$\mathbf{Claim~1}$: For any chain $C\subseteq \bigcup_{n\in \mathbb{N}} T_{u_{n}.s_{0}}$, $C$ is countable.

Suppose that there is a chain $C\subseteq \bigcup_{n\in \mathbb{N}} T_{u_{n}.s_{0}}$ being uncountable. Then there must exist an uncountable chain $C_{0} = \{v.s: v\in D_{0}\}\subseteq C$ contained in $T_{u_{n_{0}}.s_{0}}$ for some $n_{0}\in \mathbb{N}$. That is to say, for any $v.s\in C_{0}$, $v.s = \vee E_{n_{0}}.t_{c}.p_{c}$ and so $t_{c}.p_{c}$ is equal to $s$. Then $C_{0} = \{\vee E_{n_{0}}.s\}$ is a single set, which contradicts the uncountability of $C_{0}$.

$\mathbf{Claim~2}$: For any $n\in \mathbb{N}$, $B_{n}$ is Scott closed.

To this end, let $\{(m, b.s): b\in D\}$ be a non-trivial chain contained in $B_{n}$.

Case 1, $\{(m, b.s): b\in D\}\bigcap (\bigcup_{i\in \mathbb{N}\setminus \{1,2,...,n-1\}}\bigcup_{a\in E_{i}}\downarrow\!(a, u_{i}.s_{0}))$ is cofinal in $\{(m, b.s): b\in D\}$. For any $b\in D$, there are $i_{b}\in \mathbb{N}\setminus \{1,2,...,n-1\}, a_{b}\in E_{i_{b}}$ such that $(m, b.s)< (a_{b}, u_{i_{b}}.s_{0})$, where each $(a_{b}, u_{i_{b}}.s_{0})$ can be assumed to be the minimal element of $K$ by the constructions of $G$ and $\{E_{n}: n\in \mathbb{N}\}$. We set
\begin{center}
$A_{<_{r}} = \{(m, b.s): b\in D, \exists\ i_{b}\in \mathbb{N}\setminus \{1,2,...,n-1\},\ a_{b}\in E_{i_{b}}\ s.t.\ (m, b.s)<_{r} (a_{b}, u_{i_{b}}.s_{0})\}$,
\end{center}
where $<_{r}\ \in \{<_{1},\ <_{2},\ <_{3},\ <_{2};<_{1},\ <_{3};<_{1}\}$.

Case 1.1, $A_{<_{1}}$ is cofinal in $\{(m, b.s): b\in D\}$. For any $(m, b.s)\in\! A_{<_{1}}$, there are some $i_{b}\in \mathbb{N}\setminus \{1,2,...,n-1\},\ a_{b}\in E_{i_{b}}$ such that $(m, b.s)<_{1} (a_{b}, u_{i_{b}}.s_{0})$. Then $a_{b} = m$ and there is an $i_{0}\in \mathbb{N}\setminus \{1,2,...,n-1\}$ such that $(m, b.s)<_{1} (m, u_{i_{0}}.s_{0})$ from the minimality in $K$ of $(a_{b}, u_{i_{b}}.s_{0})$ for any $(m, b.s)\in A_{<_{1}}$. Thus $\sup_{b\in D}(m, b.s)\leq (m, u_{i_{0}}.s_{0})$, which means $\sup_{b\in D}(m, b.s)\in B_{n}$.

Case 1.2, $A_{<_{2}}$ or $A_{<_{2};<_{1}}$ is cofinal in $\{(m, b.s): b\in D\}$. It is easy to demonstrate that $\sup_{b\in D}(m, b.s)\in B_{n}$ in the case.

Case 1.3, $A_{<_{3}}$ is cofinal in $\{(m, b.s): b\in D\}$. For any $(m, b.s)\in A_{<_{3}}$, there are some $i_{b}\in \mathbb{N}\setminus \{1,2,...,n-1\},\ a_{b}\in E_{i_{b}}$ such that $(m, b.s)<_{3} (a_{b}, u_{i_{b}}.s_{0})$. Then $u_{i_{b}}.s_{0}\subseteq s$, which yields that $u_{i_{b}}.s_{0}$ is unique at the moment, denoted by $u_{i'}.s_{0}$, that is, for any $(m, b.s)\in A_{<_{3}}$, there exists $a_{b}\in E_{i'}$ such that $(m, b.s)<_{3} (a_{b}, u_{i'}.s_{0})$. Because $(a_{b}, u_{i'}.s_{0})\in\min K$, and we pick only a minimal element of $K$ on each level $L_{u_{n}.s_{0}}$ for any $n\in \mathbb{N}$, $(a_{b}, u_{i'}.s_{0})$ is a fixed point denoted by $(a_{0}, u_{i'}.s_{0})$. Thus $\sup_{b\in D}(m, b.s)\leq (a_{0}, u_{i'}.s_{0})$ and so $\sup_{b\in D}(m, b.s)\in B_{n}$.

Case 1.4, $A_{<_{3};<_{1}}$ is cofinal in $\{(m, b.s): b\in D\}$.  For any $(m, b.s)\in A_{<_{3};<_{1}}$, there are some $i_{b}\in \mathbb{N}\setminus \{1,2,...,n-1\},\ a_{b}\in E_{i_{b}}$ such that $(m, b.s)<_{3};<_{1} (a_{b}, u_{i_{b}}.s_{0})$, that is, $(m, b.s)<_{3} (a_{b}, s_{b})<_{1} (a_{b}, u_{i_{b}}.s_{0})$ for some $s_{b}\in \mathbb{W}^{*}$. Then $s_{b}\subseteq s$ is unique since the length of each $s_{b}$ is same as that of $u_{i_{b}}.s_{0}$, written as $s^{*}$. Assume $s = t.s^{*}$ for some $t\in \mathbb{W}^{*}$. It follows that $\min (bt)\leq a_{b}$ in light of $(m, b.s)<_{3} (a_{b}, s^*)$, where $t$ may be $\varepsilon$.

Let $i_{b_{0}}=\inf R$, where $R=\{u_{i}\mid s^*\sqsubset_{1} u_i.s_{0}\}$. Then $E_{i_b}\subseteq E_{i_{b_0}}$ for any $b\in D$ with $(m, b.s)\in A_{<_{3};<_{1}}$. Therefore, $a_{b}\in E_{i_{b_{0}}}$. If $t\neq \varepsilon$ and there is $b_{1}> \min t$, then $\min t=\min(b_1.t)\leq a_{b_{1}}$. So we have $\sup_{b\in D}(a, b.s)<_{3} (a_{b_{1}}, s^{*})<_{1} (a_{b_{1}}, u_{i_{b_{1}}}.s_{0})$. If $t = \varepsilon$ or $b\leq\min t$ for each $b\in D$, then $\sup D\leq \min t$ or $t = \varepsilon$. The fact that $\min(b.t)\leq a_b$ suggests that $\min(b.t)=b\leq a_b$, which leads to that $\sup D\leq \vee E_{i_{b_{0}}}$. Now we need to distinguish the following two cases.

Case 1.4.1, $\sup D = \vee E_{i_{b_{0}}}$. Then $\sup_{b\in D}(m, b.s) = (m, \vee E_{i_{b_{0}}}.s) = (m, \vee E_{i_{b_{0}}}.t.s^{*})$. Note that $\vee E_{i_{b_{0}}} =\sup D\leq \min t$ or $t = \varepsilon$, which yields that $\vee E_{i_{b_{0}}}.t.s^*\in T_{i_{b_{0}}.s_0}$.

Case 1.4.2, $\sup D < \vee E_{i_{b_{0}}}$. Then there exists $a'\in E_{i_{b_{0}}}$ such that $\sup D\leq a'$, which results in the conclusion that $\sup_{b\in D}(m, b.s)<_{3} (a', s^{*})<_{1} (a', u_{i_{b_{0}}}.s_{0})\in B_n$.

Case 2, $\{(m, b.s): b\in D\}\bigcap\; (\bigcup \{\downarrow L_{\vee C}: C\subseteq \bigcup_{n\in \mathbb{N}} T_{u_{n}.s_{0}}\ \mathrm{is\ a\ chain}\})$ is cofinal in $\{(m, b.s): b\in D\}$. That is to say, for any $b\in D$, there are some $m_{b}\in \mathbb{W}$ and a chain $C_{b}\subseteq \bigcup_{n\in \mathbb{N}} T_{u_{n}.s_{0}}$ such that $(m, b.s)<(m_{b}, \vee C_{b})$. Now we set
\begin{center}
$A'_{<_{r}} = \{(m, b.s): b\in D, \exists\ m_{b}\in \mathbb{W},\ C_{b}\subseteq \bigcup_{n\in \mathbb{N}} T_{u_{n}.s_{0}}\ s.t.\ (m, b.s)<_{r} (m_{b}, \vee C_{b})\}$,
\end{center}
and
\begin{center}
	$U'_{<_{r}} = \{u\in D_{0}: (m, u.s)\in A'_{<_{r}}\}$,
\end{center}
where $<_{r}\ \in \{<_{1},\ <_{2},\ <_{3},\ <_{2};<_{1},\ <_{3};<_{1}\}$.

Case 2.1, $A'_{<_{1}}$ is cofinal in $\{(m, b.s): b\in D\}$. For any $(m, b.s)\in A'_{<_{1}}$, there are $m_{b}\in \mathbb{W}$ and $C_{b}\subseteq \bigcup_{n\in \mathbb{N}} T_{u_{n}.s_{0}}$ such that $(m, b.s)<_{1}(m_{b}, \vee C_{b})$. Then $m_{b} = m$. By Claim $1$, we know that $C_b$ must be countable, which ensures the existence of the countable chain $C_{0}\subseteq \bigcup_{b\in U'_{<_{1}}} C_{b}$ being the cofinal subset of $\bigcup_{b\in U'_{<_{1}}} C_{b}$. It follows that $C_0\subseteq \bigcup_{b\in U'_{<_{1}}} C_{b}\subseteq \bigcup_{n\in \mathbb{N}} T_{u_{n}.s_{0}}$. This implies that $(m, b.s)<_{1} (m, \vee C_{0})$ for any $b\in D$. So $\sup_{b\in D}(m, b.s)\leq (m, \vee C_{0})\in B_n$.

Case 2.2, $A'_{<_{2}}$ or $A'_{<_{3}}$ is cofinal in $\{(m, b.s): b\in D\}$. For any $(m, b.s)\in A'_{<_{2}}$ or $A'_{<_{3}}$, there are some $m_{b}\in \mathbb{W}$ and $C_{b}\subseteq \bigcup_{n\in \mathbb{N}} T_{u_{n}.s_{0}}$ such that $(m, b.s)<_{2}(m_{b}, \vee C_{b})$ or $(m, b.s)<_{3}(m_{b}, \vee C_{b})$. Then $\vee C_{b}\subseteq s$, that is, $s\sqsubseteq \vee C_{b}$. Hence, $\sup_{b\in D}(m, b.s)\leq (m, \vee C_{b_0})$ for any given $b_0\in U'_{<_{2}}$ or $b_0\in U'_{<_{3}}$, which leads to that $\sup_{b\in D}(m, b.s)\in B_{n}$.

Case 2.3, $A'_{<_{2};<_{1}}$ or $A'_{<_{3};<_{1}}$ is cofinal in $\{(m, b.s): b\in D\}$. The analysis of this case is similar to Case $2.2$.

As $B_n$ is a lower set for each $n\in \mathbb{N}$, $B_{n}$ is Scott closed. Because $\{B_{n}: n\in \mathbb{N}\}$ is filtered and $\min K\bigcap B_{n}\neq\emptyset$ for any $n\in \mathbb{N}$, the set $\bigcap_{n\in \mathbb{N}}B_{n}$ meets $\min K$. Choose $(m, s)\in \bigcap_{n\in \mathbb{N}}B_{n}\bigcap\min K$. The uncountability of $A$ guarantees the existence of $u^*\in \ua \vee_{n\in \mathbb{N}}u_n$ with $u^*.s_{0}\in A$. Now we claim that there exist $u'> u^*$, $m'\in \ua \vee (\bigcup_{n\in \mathbb{N}}E_{n})=\ua \vee E_1$ such that $(m', u'.s_{0})\in \min K$. Suppose not, for any $u> u^*$, any $m\in \mathbb{W}$ with $(m, u.s_{0})\in \min K$, $m\notin {E_{1}}^{u}$. Then we could find $a_{m}\in E_{1}$ such that $a_{m}\nleqslant m$, i.e., $a_{m}> m$. It follows that
\begin{center}
$\{m\in \mathbb{W}: (m, u.s_{0})\in \min K\ \mathrm{with}\ u> u^*\}\subseteq\; \downarrow\!E_{1}$,
\end{center}
which is a contradiction since the left side of the inclusion is uncountable but the right is countable. Hence, there exists $u'> u^*$, $m'\in \ua \vee (\bigcup_{n\in \mathbb{N}}E_{n})=\ua \vee E_1$ such that $(m', u'.s_{0})\in \min K$. Now based on the above discussion, we shall know that $(m, s)\notin\ \downarrow\! L_{\vee C}$ for any chain $C\subseteq \bigcup_{n\in \mathbb{N}}T_{u_{n}.s_{0}}$. Suppose not, $(m, s)\leq (m_{0}, \vee C_{0})$ for some $m_{0}\in \mathbb{W}$ and $C_{0}\subseteq \bigcup_{n\in \mathbb{N}}T_{u_{n}.s_{0}}$. For each $c\in C_{0}$, $c$ can be written as $\vee E_{n_{c}}.t_{c}.p_{c}$ with $t_{c} = \varepsilon$ or $\vee E_{n_{c}}\leq\min t_{c},\ p_{c} = u_{n_{c}}.s_{0}$ or $p_{c}\sqsubset_{1} u_{n_{c}}.s_{0}$. Note that there is a cofinal subset $C'_{0}$ of $C_{0}$ of the form $\{m_{c}.\hat{s}: c\in C'_{0}\}$, in which $m_{c} = \vee E_{n_{c}},\ \hat{s} = t_{c}.p_{c}\triangleq t.p$ for each $c\in C'_{0}$, $\bigvee (\bigcup_{c\in C'_{0}} E_{n_{c}})\leq \min t$ or $t = \varepsilon$. According to $\bigvee (\bigcup_{c\in C'_{0}} E_{n_{c}})\leq \vee E_{1}\leq m'$, we can get that
\begin{center}
$(m, s) \leq (m_{0}, \vee C_{0}) = (m_{0}, \vee C'_{0})\leq (m_{0}, \bigvee (\bigcup_{c\in C'_{0}} E_{n_{c}}).t.p)<_{3} (m', p)<_{1} (m', u'.s_{0})$.
\end{center}
This violates the assumption that both $(m, s)$ and $(m', u'.s_{0})$ are  minimal elements of $K$. As a result, for $(m, s)\in\bigcap_{n\in \mathbb{N}}B_{n}\bigcap\min K$, $(m, s)\in \bigcap_{n\in \mathbb{N}}(\bigcup_{i\in \mathbb{N}\setminus \{1,2,...,n-1\}}\bigcup_{a\in E_{i}}\downarrow\!(a, u_{i}.s_{0}))$. But this is impossible with a similar reason in the proof of~Lemma \ref{min2}.
\end{proof}
\end{lemma}

\begin{lemma}\label{min4}
If $K$ is a compact saturated subset of $\mathcal{Z}$, then $\min K$ is countable.
\begin{proof}
By way of contradiction, suppose that $\min K$ is uncountable. In virtue of Lemma \ref{min2}, we have $\min K\cap L_s$ is countable for any $s\in \min K$. The fact that $\min K=\bigcup_{s\in f(\min K)}(\min K\cap L_s)$ induces that $f(\min K)$ is uncountable, where the map $f: \mathcal{Z}\rightarrow \mathcal{T}$ is defined by $f(m, s) = s$. In addition, the continuity of $f$ reveals that $f(\min K)$ is a compact subset of $\mathcal{T}$. Then $\min f(\min K)$ is finite due to Lemma \ref{min1}, which yields that there is $s\in \min f(\min K)$ such that $\uparrow\!s\cap f(\min K)$ is uncountable. Since the length of $s$ is finite, there exists $s_{0}\geq s$ such that $\min K\bigcap (\bigcup_{u\in \mathbb{W}}L_{u.s_{0}})$ is uncountable. This contradicts the result of Lemma \ref{min3}. Hence, $\min K$ is countable.
\end{proof}
\end{lemma}

\begin{theorem}
$\Sigma \mathcal{Z}$ is well-filtered.
\begin{proof}
Recall that a $T_0$ space if well-filtered if and only if all its $K\!F$-sets are closures of singletons. We assume that $\Sigma \mathcal{Z}$ is not well-filtered. 
Then there exists a $K\!F$-set $\downarrow\!L_{s_{0}}$ for some $s_{0}\in \mathbb{W}^{*}$ by Theorem~\ref{A} and that all $K\!F$-sets are irreducible. This implies that there is a filtered family $\{Q_{i}: i\in I\}$ contained in $Q(\mathcal{Z})$ such that $\downarrow\!\!L_{s_{0}}$ is a minimal closed set that intersects each $Q_{i}$. Set $K_{i} = \uparrow\!(Q_{i}\ \cap \downarrow\!L_{s_{0}})$, then a straightforward verification establishes that $\{K_{i}: i\in I\}\subseteq_{flt}Q(\mathcal{Z})$ and $\downarrow\!L_{s_{0}}$ is still a minimal closed set that intersects each $K_{i}$. Note that $\min K_{i}\subseteq\ \downarrow\!L_{s_{0}}$ for each $i\in I$. Now we fix $j_{1}\in I$ and let $J = \uparrow\!j_{1}\cap I$, then $J$ is a cofinal subset of $I$ and $\downarrow\!L_{s_{0}}$ is a minimal closed set that intersects each $K_{j}, j\in J$. The assumption that $\downarrow\!L_{s_{0}}\cap K_{j}\neq\emptyset$ implies that  $\downarrow\!L_{s_{0}}\cap \min K_{j}\neq\emptyset$ for all $j\in J$. Pick $(m_{j}, s_{j})\in\ \downarrow\!L_{s_{0}}\cap \min K_{j}$ for each $j\in J$.

$\mathbf{Claim~1}$: There exists $j_{0}\in J$ such that $s_{j} = s_{0}$ for any $j\geq j_{0}$.

Assume for the sake of a contradiction that for any $j\in J$, there is $i\geq j$ such that $s_{i}\neq s_{0}$, that is, $s_{i}\sqsubset s_{0}$. Let $J_{1} = \{j\in J: s_{j}\sqsubset s_{0}\}$. Then $J_{1}$ is cofinal in $J$ and so $\downarrow\!L_{s_{0}}$ is a minimal closed set that intersects each $K_{j}, j\in J_{1}$. The minimality of $\da L_{s_0}$ infers that $\overline{\{(m_{j}, s_{j}): j\in J_{1}\}} = \downarrow\!L_{s_{0}}$. This indicates that $\{(m_{j}, s_{j}): j\in J_{1}\}$ is uncountable. If not, $\{(m_{j}, s_{j}): j\in J_{1}\}$ is countable, we set $B = \bigcup_{j\in J_{1}}\downarrow\!(m_{j}, s_{0})\bigcup\ \downarrow\!(\vee_{j\in J_1} m_{j}, s_{0})$. Then it is easy to verify that $B$ is Scott closed and $B\cap K_{j}\neq\emptyset, j\in J_{1}$, but $B\subsetneqq \downarrow\!L_{s_{0}}$, which is a contradiction to the minimality of $\downarrow\!L_{s_{0}}$. The above discussion also reveals that $\{m_{j}: j\in J_{1}\}$ is uncountable. Now we let
\begin{center}
$B_{j} = \{(m_{i}, s_{i}): i\geq j, (m_{i}, s_{i})\in \min K_{j}\}, \forall j\in J_{1}$,
\end{center}
then it follows that each $B_{j}$ is countable by Lemma \ref{min4}. Note that for any given $j\in J_1$, $I_{j}=\{i\in J_i: i\geq j\}$ is a cofinal subset of $I$. This means that $\overline{\{(m_{r}, s_{r}): r\in I_{j}\}} = \downarrow\!L_{s_{0}}$, which yields that $\{m_{r}: r\in I_{j}\}$ is uncountable for a similar reason as $J_1$.
Therefore,
\begin{center}
$B'_{j} = \{(m_{i}, s_{i}): i\geq j, (m_{i}, s_{i})\notin \min K_{j}\}$
\end{center}
is uncountable for all $j\in J_{1}$, and $\{m_{i}: (m_{i}, s_{i})\in B'_{j}\}$ is uncountable. It turns out that $\overline{B'_j}=\da L_{s_0}$. Given $j\in J_1$, for any $(m_{i}, s_{i})\in B'_{j}$, there is $(n_{i_{k}}, a_{i_{k}})\in \min K_{j}$ such that $(n_{i_{k}}, a_{i_{k}})< (m_{i}, s_{i})$. The countability of $\min K_{j}$ ensures the existence of $(n_{i_{0}}, a_{i_{0}})\in \min K_{j}$ with  $\uparrow\!\!(n_{i_{0}}, a_{i_{0}})\bigcap B'_{j}$ and $E_{i_0} = \{m_{i}: (m_{i}, s_{i})\in\ \uparrow\!\!(n_{i_0}, a_{i_0})\bigcap B'_{j}\}$ being uncountable, by the aids of the equation $\{m_{i}: (m_{i}, s_{i})\in B'_{j}\} = \bigcup_{i\geq j}\{m_{i}: (m_{i}, s_{i})\in\ \uparrow\!(n_{i_{k}}, a_{i_{k}})\bigcap B'_{j}\}$ and the uncountability of the set $\{m_{i}: (m_{i}, s_{i})\in B'_{j}\}$.
Thus there is a $m_{i}\in E_{i_0}$ with $m_{i}\neq n_{i_0}$. In a conclusion, we can find $(m_{i_{j}}, s_{i_{j}})\in B'_{j}$ and $(n_{i_{j}}, a_{i_{j}})\in \min K_{j}$ satisfying $(n_{i_{j}}, a_{i_{j}})< (m_{i_{j}}, s_{i_{j}})$, with $n_{i_{j}}\neq m_{i_{j}}$ for any $j\in J_{1}$, which means that $a_{i_{j}}\sqsubset_{2} s_{i_{j}}$ or $a_{i_{j}}\sqsubset_{2};\sqsubset_{1} s_{i_{j}}$. The fact that $\overline{\{(n_{i_{j}}, a_{i_{j}}): j\in J_{1}\}}\bigcap K_{j}\neq \emptyset$, for any $ j\in J_{1}$ suggests that $\overline{\{(n_{i_{j}}, a_{i_{j}}): j\in J_{1}\}} = \downarrow\!L_{s_{0}}$ from the minimality of $\da L_{s_0}$. This implies that $\overline{\{a_{i_{j}}: j\in J_{1}\}} = \overline{\{s_{0}\}} =\ \downarrow\!\!s_{0}$. It follows that $s_{0} = \vee_{j\in J_{1}}a_{i_{j}}$. The fact that $a_{i_{j}}\in f(\min K_{j})\subseteq\ f(K_{j_{1}})\subseteq\ \uparrow\!\min f(\min K_{j_{1}}) =\ \uparrow\!F$ indicates that $\{a_{i_{j}}: j\in J_{1}\} = \bigcup_{a\in F}(\uparrow\!a\cap \{a_{i_{j}}: j\in J_{1}\})$, where $F = \min f(\min K_{j_{1}})$. Then $F$ is finite by Lemma \ref{min1}, which leads to 
\begin{center}
$\overline{\{a_{i_{j}}: j\in J_{1}\}}\subseteq \bigcup_{a\in F}\downarrow\!\vee (\uparrow\!a\cap \{a_{i_{j}}: j\in J_{1}\})$.
\end{center}
It follows that $s_{0}\in \downarrow\!\vee (\uparrow\!a_{0}\cap \{a_{i_{j}}: j\in J_{1}\})$ for some $a_{0}\in F$, that is, $s_{0}\leq \vee (\uparrow\!a_{0}\cap \{a_{i_{j}}: j\in J_{1}\})$. Then $s_{0} = \vee (\uparrow\!a_{0}\cap \{a_{i_{j}}: j\in J_{1}\})$ according to $s_{0} = \vee_{j\in J_{1}}a_{i_{j}}$. It turns out that there is a chain of the form $\{v.s: v\in D\}$ being a cofinal subset of $\uparrow\!a_{0}\cap \{a_{i_{j}}: j\in J_{1}\}$ enjoying $s_{0} = \sup_{v\in D}v.s$. This means that $|s_{0}| = |v.s|$ or $|s_{0}|+1 = |v.s|$. Note that $v.s\in (\uparrow\!a_{0}\cap \{a_{i_{j}}: j\in J_{1}\})$ for any $v\in D$. Then assume $v.s=a_{i_{j_v}}$. We notice that $v.s\sqsubset_{2} s_{i_{j_v}}$ or $v.s\sqsubset_{2};\sqsubset_{1} s_{i_{j_v}}$, and conclude that $|v.s|> |s_{i_{j_v}}|$. We claim that $|s_{i_{j_v}}| = |s_{0}|$. If not, $|s_{i_{j_v}}|\neq |s_{0}|$. The fact that $s_{i_{j_v}}\sqsubset s_{0}$ deduces $|s_{i_{j_v}}|>|s_{0}|$, then $|v.s|> |s_{i_{j_v}}|> |s_{0}|$, that is, $|v.s|> |s_{0}|+1$, a contradiction. Hence, $|s_{i_{j_v}}| = |s_{0}|$, so $|v.s|> |s_{0}|$. This implies that $D$ is uncountable because $s_{0} = \sup_{v\in D}v.s$, which in turn implies that $s_{0} = s$, and it means $v.s_{0}\sqsubset_{2} s_{i_{j_v}}$ or $v.s_{0}\sqsubset_{2};\sqsubset_{1} s_{i_{j_v}}$. Then $s_{0}\sqsubseteq s_{i_{j_v}}$. But this is a contradiction to the fact that $s_{i_{j_v}}\sqsubset s_{0}$. Therefore, Claim $1$ holds.

$\mathbf{Claim~2}$: There exists $j_{2}\in J$ such that $\min K_{j_{2}}\subseteq L_{s_{0}}$.

Assume that for any $j\in J$, $\min K_{j}$ is not contained in $L_{s_{0}}$. Then $\min K_{j}\bigcap\ (\downarrow\! L_{s_{0}}\!\setminus L_{s_{0}})\neq\emptyset$ from the fact that $\min K_{j}\subseteq\downarrow\!L_{s_{0}}$ for any $ j\in J$. We pick $(m_{j}, s_{j})\in \min K_{j}\bigcap( \downarrow\! L_{s_{0}}\setminus\! L_{s_{0}})$, then $s_{j}\sqsubset s_{0}$ for all $j\in J$. This contradicts Claim $1$.

$\mathbf{Claim~3}$: $\min K_{j_{2}}$ is finite.

We proceed by contradiction. Suppose that $\min K_{j_{2}}$ is infinite. By Lemma \ref{min4}, we know that $\min K_{j_{2}}$ is countable. Let $\min K_{j_{2}} = \{(m_{n}, s_{0}): n\in \mathbb{N}\}$, where $\mathbb{N}$ denotes all positive natural numbers. Then we can identify $m'\in \mathbb{W}$, which is an upper bound of $\{m_n:n\in \mathbb{N}\}$.

For any $n\in \mathbb{N}$, set
\begin{center}
$B_{n} = \bigcup_{i\in \mathbb{N}\setminus\{1,2,...,n-1\}}\downarrow\!(m_{i}, s_{0})\bigcup\downarrow\!(m', s_{0})$.
\end{center}
It is easy to check that $B_{n}$ is Scott closed for any $n\in \mathbb{N}$. By the construction of $B_n$, we have $\min K_{j_2}\cap B_n\neq \emptyset$ for each $n\in \mathbb{N}$. Then the compactness of $\min K_{j_{2}}$ implies that $\bigcap_{n\in \mathbb{N}}B_{n}\bigcap \min K_{j_{2}}\neq\emptyset$. Choose $(m_{n_0},s_0)\in \bigcap_{n\in \mathbb{N}}B_{n}\bigcap \min K_{j_{2}}$. It follows that $(m_{n_0},s_0)\in B_{n_0+1}$. Now the remaining arguments are similar to that in Lemma \ref{min2}.

$\mathbf{Claim~4}$: For any $j\geq j_{2}, \min K_{j}\subseteq \min K_{j_{2}}$.

By way of contradiction, assume that there is a $j'\geq j_{2}$ such that $\min K_{j'}\not\subseteq \min K_{j_{2}}$, that is, there exists $(m, s)\in \min K_{j'}\setminus \min K_{j_{2}}$. It follows that there is $(m_{1}, s_{1})\in \min K_{j_{2}}$ such that $(m, s)> (m_{1}, s_{1})$, which implies $s_{1}\sqsubset s$. The fact that $\min K_{j_{2}}\subseteq L_{s_{0}}$ suggests that $s_{1} = s_{0}$. It indicates that $s_{0}\sqsubset s$, which contradicts $s\sqsubseteq s_{0}$ as $\min K_{j'}$ is contained in $ \downarrow\!L_{s_{0}}$.

Now let $J' = \{j\in J: j\geq j_{2}\}$, which is a cofinal subset of $J$, and so $\downarrow\!L_{s_{0}}$ is a minimal closed set intersects $K_{j}$, for any $j\in J'$. Assume $\min K_{j_{2}} = \{(a_{i}, s_{0}): i = 1,2,...,n\}$ and set $F = \bigcup_{i\in \{1,2,...,n\}}\downarrow\!(a_{i}, s_{0})$. Then $F$ is Scott closed. Note that $F\cap K_{j}\neq\emptyset$ for any $ j\in J'$ from Cliam 4. But $F\subsetneqq\ \downarrow\!L_{s_{0}}$. This violates the minimality of $\downarrow\!L_{s_{0}}$. So our assumption that  $\Sigma \mathcal{Z}$ is not well-filtered must have been wrong. 
\end{proof}
\end{theorem}

\subsection{$\mathbf{WF}$ is not $\Gamma$-faithful}
In this subsection, we will show that $\mathbf{WF}$ is not $\Gamma$-faithful by using the well-filtered dcpo $\mathcal{Z}$ constructed above.  
It is trivial to check that $\Sigma \mathcal{Z}$ and $\Sigma (I\!R\!R(\mathcal{Z}))\triangleq \Sigma \hat{\mathcal{Z}}$ are not homeomorphism. Thus if we derive $\Gamma \mathcal{Z}\cong \Gamma\hat{\mathcal{Z}}$ and $\Sigma \hat{\mathcal{Z}}$ is well-filtered, then this suffices to show that the category $\mathbf{WF}$ is not $\Gamma$-faithful.

\begin{theorem}
$\Gamma\mathcal{Z}\cong \Gamma\hat{\mathcal{Z}}$, that is, the family of  Scott closed subsets of $\mathcal{Z}$ and $\hat{\mathcal{Z}}$ are isomorphic.
\begin{proof}
One sees clearly that the Scott topology of $\mathcal{Z}$ is isomorphic to the lower Vietoris topology of $\hat{\mathcal{Z}}$, which is the family $\tau=\{\lozenge U: U\in \sigma(\mathcal{Z})\}$, where $\lozenge U = \{A\in \hat{\mathcal{Z}}: A\cap U\neq\emptyset\}$. Note that $\hat{\mathcal{Z}}$ endowed with the lower Vietoris topology is well-filtered. So it remain to prove that the Scott topology of $\hat{\mathcal{Z}}$ and the lower Vietoris topology of $\hat{\mathcal{Z}}$ coincide. It is clear to see that each closed set in the lower Vietoris topology is Scott closed. Now we show the converse. To this end, for each $A\in \Gamma(\mathcal{Z})$, we write $\Box A = \{B\in \hat{\mathcal{Z}}: B\subseteq A\}$ and choose $\mathcal{A}\in \Gamma(\hat{\mathcal{Z}})$. Then it suffices to show that $\bigcup\mathcal{A}$ is Scott closed and $\mathcal{A} = \Box(\bigcup\mathcal{A})$. 

First, by Proposition 2.2 (4) in \cite{hzp}, we know that for the dcpo $\mathcal{Z}$, $\bigcup\mathcal{A}$ is Scott closed. Next, we confirm that $\mathcal{A} = \Box(\bigcup\mathcal{A})$. That $\mathcal{A} \subseteq \Box(\bigcup\mathcal{A})$ is trivial. Conversely, let $A\in \Box(\bigcup\mathcal{A})$. Then $A\subseteq \bigcup\mathcal{A}$. By Theorem \ref{A}, we know that $A\in \{\da L_s:s\in \mathbb{W}^*\}\cup \{\da(m,s):(m,s)\in \mathbb{W}\times \mathbb{W}^*\}$. If $A$ is the form $\downarrow\!(m, s)$ for some $(m,s)\in \mathbb{W}\times \mathbb{W}^*$, then $A\in \mathcal{A}$ evidently. Else, $A=\ \downarrow\!L_{s}$ for some $s\in \mathbb{W}^{*}$. It follows that for each $m\in \mathbb{W}$, there is $A_{m}\in \mathcal{A}$ such that $(m, s)\in A_{m}$. It turns out that $\downarrow\!(m, s)\in \mathcal{A}$ since $\mathcal{A}$ is a lower set. Furthermore, we get that $\downarrow\!L_{m.s}\in \mathcal{A}$ for each $m\in \mathbb{W}$ as $\downarrow\!L_{m.s}\subseteq\ \downarrow\!(m, s)$. Note that $\{\downarrow\!L_{m.s}: m\in \mathbb{W}\}$ is a directed family contained in $\mathcal{A}$ and $\mathcal{A}$ is Scott closed. Then $A=\downarrow\!L_{s} = \sup_{m\in \mathbb{W}}\downarrow\!L_{m.s}\in \mathcal{A}$. So $\Box (\bigcup\mathcal{A})\subseteq \mathcal{A}$. Indeed, we have proved that $\Gamma\mathcal{Z}\cong \Gamma\hat{\mathcal{Z}}$.
\end{proof}
\end{theorem}

\section{The category of weak dominated dcpo's is $\Gamma$-faithful}

Finally, we give a $\Gamma$-faithful category, that of \emph{weakly dominated dcpo's}, which is strictly larger than the category of dominated dcpo's. 
\begin{definition}\cite{lost}
	Let $P$ be a poset and $x,y\in P$. We say that $x$ is \emph{beneath} y, denoted by $x\prec^{*} y$, if for every nonempty Scott closed set $C\subseteq P$ for which $\sup C$ exists, the relation $y\leq \sup C$ always implies that $x\in C$. An element $x$ of a poset $P$ is called \emph{$C$-compact} if $x\prec^{*}x$. If $P=\Gamma L$ for some poset $L$, then the set of all $C$-compact elements of $\Gamma L$ is denoted by $C(\Gamma L)$.
\end{definition}

The following proposition is a corollary of \cite[Proposition~3.4]{lost}.
\begin{proposition}\label{subdcpo100}
	$C(\Gamma L)$ is a subdcpo of $\Gamma L$.
\end{proposition}
\begin{corollary}\cite{lost}\label{principle}
	The family of sets $\{\da x: x\in L\}$ is a subset of $C(\Gamma L)$.
\end{corollary}
The following results is crucial for further discussion.
\begin{lemma}\label{a2}
	Let $L$ be a poset. Then $\bigcup \mathcal{A}$ is a Scott closed subset of $L$ for each $\mathcal{A}\in \Gamma(C(\Gamma L))$.
	\begin{proof}
		One sees directly that $\bigcup\mathcal{A}$ is a lower set since each element in $\mathcal{A}$ is a lower set. Let $D$ be a directed subset of $\bigcup\mathcal{A}$. Then there exists $A_{d}\in \mathcal{A}$ such that $d\in A_{d}$ for any $d\in D$. It follows that $\da d\subseteq A_{d}$ because $A_d$ is a lower set. From Corollary \ref{principle}, we know that $\da d\in C(\Gamma L)$ for any $d\in D$. The fact that $\mathcal{A}$ is a lower set in $C(\Gamma L)$ implies that $\da d\in \mathcal{A}$. This means that $\{\da d: d\in D \}$ is a directed subset of $\mathcal{A}$, which yields that $\sup_{d\in D}\da d=\da \sup D\in \mathcal{A}$. Hence, $\sup D\in \bigcup\mathcal{A}$.
	\end{proof}
\end{lemma}

\begin{definition}
	Let $L$ be a poset. A set $B$ is said to be a \emph{$C$-compact set}, if $cl(B)$ is a $C$-compact element of $\Gamma L$.  
\end{definition}

C-compact sets are preserved by Scott-continuous functions. 
\begin{lemma}\label{ccompact}
	Let $L,P$ be two posets, if the function $f:L\rightarrow P$ is Scott continuous, then $f(A)$ is a $C$-compact set of $P$ for any $A\in C(\Gamma L)$. 
	\begin{proof}
We prove that $cl(f(A))\prec^{*}cl(f(A))$. To this end, let $\mathcal{A}\in \Gamma(\Gamma P)$ with $cl(f(A))\subseteq \bigcup \mathcal{A}$. Then $A\subseteq \bigcup_{B\in \mathcal{A}}f^{-1}(B)$. We write $\mathcal{B}=\da \{f^{-1}(B): B\in \mathcal{A}\}$. 
		
		We claim that $\mathcal{B}$ is a Scott closed set of $\Gamma L$. Obviously, $\mathcal{B}$ is a lower set. Let $(A_{i})_{i\in I}$ be a directed subset of $\mathcal{B}$. Then we know that there is  $B_{i}\in \mathcal{A}$ with $A_{i}\subseteq f^{-1}(B_{i})$ for each $i\in I$. It follows that $f(A_{i})\subseteq B_{i}$, which means that $cl(f(A_{i}))\subseteq B_{i}\in \mathcal{A}$ for any $i\in I$. The fact that $\mathcal{A}$ is a lower set in $\Gamma P$ tells us that $cl(f(A_{i}))\in \mathcal{A}$ for each $i\in I$. Then we have $cl(f(A_{i}))_{i\in I}$ is a directed subset of $\mathcal{A}$. As $\mathcal{A}$ is Scott closed,  it follows  that $\sup_{i\in I}cl(f(A_{i}))=cl(f(\bigcup_{i\in I}A_{i}))\in \mathcal{A}$. Thus, $\sup_{i\in I}A_{i}=cl(\bigcup_{i\in I}A_{i})\in \da \{f^{-1}(B)\mid B\in \mathcal{A}\}$.
		
		It is easy to see that $A\subseteq \sup \mathcal{B}=\bigcup \mathcal{B}$. Then the assumption that $A$ is a $C$-compact element of $\Gamma L$ guarantees the existence of $B\in \mathcal{A}$ such that $A\subseteq f^{-1}(B)$. So we conclude $cl(f(A))\subseteq B\in \mathcal{A}$. Again, as $\mathcal{A}$ is Scott closed, hence a lower set, we obtain that $cl(f(A))\in \mathcal{A}$. As a result, $f(A)$ is a $C$-compact set of $P$. 
	\end{proof}
\end{lemma}
\begin{lemma}\label{a3}
	Let $L$ be a poset and $\mathcal{B}\in \Gamma(C(\Gamma L))$. If $\mathcal{B}$ is closed under suprema of $C$-compact sets, then $\mathcal{B}=\Box \bigcup\mathcal{B}=\{A\in C(\Gamma L): A\subseteq \bigcup\mathcal{B}\}$.
	\begin{proof}
That $\mathcal{B}\subseteq \Box \bigcup\mathcal{B}$ is trivial. Conversely, for any $A\in \Box \bigcup\mathcal{B}$, similar to the proof of Lemma \ref{a2}, we know that $\{\da x: x\in A\}$ is a subset of $\bigcup\mathcal{B}$. Then the function $\eta=(x\mapsto \da x):L\rightarrow C(\Gamma L)$ is well-defined and Scott continuous. Then $\eta(A)=\{\da x: x\in A\}$ is a $C$-compact set of $C(\Gamma L)$ by Lemma \ref{ccompact}. The assumption that $\mathcal{B}$ is closed under suprema of $C$-compact sets ensures  that $A=\sup_{x\in A}\da x=\sup \eta(A)\in \mathcal{B}$.
	\end{proof}
\end{lemma}
By applying Lemma \ref{a2} and Lemma \ref{a3}, we have the following lemma. 
\begin{lemma}
	For arbitrary dcpo's D and E, $C(\Gamma D)\cong C(\Gamma E)$ if and only if  $\Gamma D\cong \Gamma E$.
\end{lemma}
\begin{proof}
Assume that $g \colon C(\Gamma D)\to C(\Gamma E)$ is an isomorphism. Then the map $g' \colon \Gamma D\to \Gamma E$ that sends a Scott closed subset $A$ of $D$ to 
$\bigcup\{ g(K) \mid K\in C(\Gamma D)~\text{and}~K\subseteq A \}$ witnesses an isomorphism between $\Gamma D$ and $\Gamma E$. 
\end{proof}

\begin{definition}
	Given $A,B\in C(\Gamma P)$, we write \emph{$A\lhd B$} if there is $x \in B$ such that $A\subseteq \da x$. We write $\nabla B$ for the set $\{A\in C(\Gamma P): A\lhd B\}$.
\end{definition}
\begin{lemma}\label{a0}
	Let $L$ be a dcpo and $A \in C(\Gamma L)$.
	
	$(1)$ $A = \sup \nabla A$.
	
	$(2)$ $\nabla A$ is a $C$-compact set of $C(\Gamma L)$.
	\begin{proof}
		The first statement is trivial because $\da x\lhd A$ for any $x\in A$. 
		
		For the second statement, it suffices to prove that $cl(\nabla A)\prec^{*}cl(\nabla A)$. To this end, let $\mathcal{A}^{2}\in \Gamma(\Gamma(C(\Gamma L)))$ with $cl(\nabla A)\subseteq \bigcup \mathcal{A}^{2}$.
		
		$\mathbf{Claim~1}$: $A\subseteq \bigcup_{\mathcal{A}\in \mathcal{A}^{2}}\bigcup \mathcal{A}$.
		
		For any $x\in A$, $\da x\in \nabla A\subseteq \bigcup \mathcal{A}^{2}$. Then there is $\mathcal{A}_{x}\in \mathcal{A}^{2}$ such that $\da x\in \mathcal{A}_{x}$ for each $x\in A$. It follows that $x\in \bigcup\mathcal{A}_x$, and we have $A\subseteq \bigcup_{\mathcal{A}\in \mathcal{A}^{2}}\bigcup \mathcal{A}$.
		
		$\mathbf{Claim~2}$: $\mathcal{B}=\da \{ \bigcup \mathcal{A}\mid \mathcal{A}\in \mathcal{A}^{2}\}\in \Gamma(\Gamma L)$.
		
		It is obvious that $\mathcal{B}$ is a lower set. Let $(B_{i})_{i\in I}$ be a directed subset of $\mathcal{B}$. Then there exists $\mathcal{A}_{i}\in \mathcal{A}^{2}$ with $B_{i}\subseteq \bigcup\mathcal{A}_{i}$ for every $i\in I$. By Corollary \ref{principle}, we know that $\{\da x\mid x\in B_{i}\}\subseteq \mathcal{A}_{i}$ for every $i\in I$. It follows that the Scott closure $\mathcal{B}_{i}=cl(\{\da x\mid x\in B_{i}\})$, taken inside $C(\Gamma L)$, is contained in $\mathcal{A}_{i}$. Then $\mathcal{B}_{i}\in \mathcal{A}^{2}$ as $\mathcal{A}^{2}$ is a lower set. This deduces that $\sup_{i\in I}\mathcal{B}_{i}=cl(\bigcup_{i\in I}\mathcal{B}_{i})\in \mathcal{A}^{2}$ since $(\mathcal{B}_i)_{i\in I}$ is a directed subset of $\mathcal{A}^2$. It turns out that $cl(\bigcup_{i\in I}B_{i})\subseteq \bigcup cl(\bigcup_{i\in I}\mathcal{B}_{i})$. Therefore, $\sup_{i\in I}B_{i}=cl(\bigcup_{i\in I}B_{i})\in \mathcal{B}$. In a conclusion, $\mathcal{B}$ is Scott closed in $\Gamma L$.
		
		Since $A$ is a $C$-compact element of $\Gamma L$ and $A\subseteq \bigcup \mathcal{B}$, that $A\in \mathcal{B}$ holds. This means that there exists $\mathcal{A}\in \mathcal{A}^{2}$ such that $A\subseteq \bigcup \mathcal{A}$. It turns out that $\{\da x\mid x\in A\}\subseteq \mathcal{A}$, which yields that $\nabla A\subseteq \mathcal{A}$. So $cl(\nabla A)\in \mathcal{A}^{2}$.
	\end{proof}
\end{lemma}
\begin{definition}
	A dcpo $L$ is called \emph{$C$-compactly complete} if every $C$-compact subset of $L$ has a supremum.
\end{definition}
\begin{proposition}\label{a1}
	Let $L$ be a dcpo. Then $(C(\Gamma L),\subseteq )$ is $C$-compactly complete. Especially, $\sup \mathcal{A}=\bigcup \mathcal{A}$ for any $\mathcal{A}\in C(\Gamma (C(\Gamma L)))$.
	\begin{proof}
		It suffices to prove that $\bigcup \mathcal{A}\in C(\Gamma L)$. To this end, let $\mathcal{B}\in \Gamma(\Gamma L)$ with $\bigcup \mathcal{A}\subseteq \bigcup \mathcal{B}$. Then $A\subseteq \bigcup\mathcal{B}$ for any $A\in \mathcal{A}$. The fact that $A$ is a $C$-compact element of $\Gamma L$ implies that $A\in \mathcal{B}$. This means that $\mathcal{A}\subseteq \bigcup_{B\in \mathcal{B}}\Box B$, where $\Box B=\{D\in C(\Gamma L)\mid D\subseteq B\}$. Ones sees obviously that $\Box B$ is a Scott closed set of $C(\Gamma L)$.
		
		Now we claim that $\mathcal{B}^{2}=\da \{\Box B\mid B\in \mathcal{B}\}$ is a Scott closed set of $\Gamma(C(\Gamma L))$. Clearly, $\mathcal{B}^{2}$ is a lower set. Let $(\mathcal{A}_{i})_{i\in I}$ be a directed subset of $\mathcal{B}^{2}$. Then there exists $B_{i}\in \mathcal{B}$ with $\mathcal{A}_{i}\subseteq \Box B_{i}$ for any $i\in I$. This means that $\bigcup\mathcal{A}_{i}\subseteq B_{i}$ for any $i\in{I}$. By Lemma~\ref{a2}, we know that the set $\bigcup\mathcal{A}_i$ is Scott closed for $i\in I$, which leads to that $\bigcup\mathcal{A}_{i}\in \mathcal{B}$,  as $\mathcal{B}$ is Scott closed. It follows that $(\bigcup\mathcal{A}_{i})_{i\in I}$ is a directed subset of $\mathcal{B}$. So $\sup _{i\in I}\bigcup \mathcal{A}_{i}=cl(\bigcup_{i\in I}\bigcup \mathcal{A}_{i})\in \mathcal{B}$. This implies that $\sup_{i\in I}\mathcal{A}_i\subseteq \Box (cl(\bigcup_{i\in I}\bigcup \mathcal{A}_{i}))$.
		So, $\sup_{i\in I}\mathcal{A}_{i}\in\mathcal{B}^{2}$.
		
		Note that $\mathcal{A}\subseteq \bigcup \mathcal{B}^{2}$. Then $\mathcal{A}\in \mathcal{B}^{2}$ from the fact that $\mathcal{A}\in C(\Gamma(C(\Gamma L)))$. So there exists $B\in \mathcal{B}$ with $\mathcal{A}\subseteq \Box B$, and then  $\bigcup \mathcal{A}\subseteq B$. Thus $\bigcup \mathcal{A}\in \mathcal{B}$ by Scott closedness of $\mathcal{B}$.
	\end{proof}
\end{proposition}
Finally, we introduce the category of all weakly dominated dcpo's which serves as a $\Gamma$-faithful category.
\begin{definition}
	A dcpo $L$ is called \emph{weakly dominated} if for every $A\in C(\Gamma L)$, the collection $\nabla A$ is Scott closed in $C(\Gamma L)$.
\end{definition}

\begin{definition}
	Let $L$ be a $C$-compactly complete dcpo and $x,y\in L$. We write $x\prec y$ if for all closed $C$-compact subsets $A$, $y\leq \sup A$ implies that $x\in A$. We say that $x\in L$ is $\prec$-compact if $x\prec x$, and denote the set of all $\prec$-compact elements by $K(L)$.
\end{definition}

\begin{proposition}
	$A\lhd B$ implies that $A\prec B$ for any $A,B\in C(\Gamma L)$.
	\begin{proof}
		One sees immediately by Proposition \ref{a1}.
	\end{proof}
\end{proposition}
\begin{lemma}
	A dcpo $L$ is weakly dominated, if and only if $B\prec A$ implies $B\lhd  A$ for all $A,B\in C(\Gamma L)$.
	\begin{proof}
		If $L$ is weakly dominated, then $\nabla A\in C(\Gamma(C\Gamma(L)))$ by Lemma \ref{a0}.  Assume $B\prec A$, then $B \prec A\subseteq \bigcup \nabla A$. This means that $B\in \nabla A$.
		
		Conversely, if $B\prec A$ implies $B\lhd A$ for all $A,B\in C(\Gamma L)$. We need to check that $\nabla A$ is Scott closed in $C(\Gamma L)$.  The set $\nabla A$ is obviously  a lower set. Let $(B_{i})_{i\in I}$ be a directed subset of $\nabla A$. Then it remains to show that $\sup_{i\in I}B_{i}\prec A$ by the assumption that $B\prec A$ implies $B\lhd A$. For any $\mathcal{A}\in C(\Gamma(C(\Gamma L)))$, if $A\subseteq \bigcup\mathcal{A}$, then hat $B_{i}\prec A$ implies  that $B_{i}\in \mathcal{A}$. Note that $\mathcal{A}$ is a Scott closed set of $C(\Gamma L)$ and $(B_{i})_{i\in I}$ is a directed subset of $\mathcal{A}$. We conclude that $\sup_{i\in I}B_{i}=cl(\bigcup_{i\in I}B_{i})\in \mathcal{A}$.
	\end{proof}
\end{lemma}
By observing the above results, we can get the following corollaries.
\begin{proposition}
	For a weakly dominated dcpo $L$, the only $\prec$-compact elements of $C(\Gamma L)$ are exactly the
	principal ideals $\da x$, $x \in L$. 
\end{proposition}

\begin{theorem}
	Let $L$ and $P$ be weakly dominated dcpo's. The following are equivalent:
\begin{enumerate}
\item $L\cong P$; 
\item  $\Gamma L\cong \Gamma P$;
\item $C(\Gamma L)\cong C(\Gamma P)$.  \hfill $\Box$
\end{enumerate}
\end{theorem}

The following theorem follows immediately. 

\begin{theorem}
	The category of all weakly dominated dcpo's is $\Gamma$-faithful. \hfill $\Box$
\end{theorem}

\section{Weakly dominated dcpo's may fail to be dominated}

From the definitions of dominated dcpo's and weakly dominated dcpo's and the fact that $C(\Gamma L)$ is contained in $IRR(L)$ for each dcpo $L$, we could easily see that   dominated dcpo's are weakly dominated. We show in this section, by giving a concrete example, that weak dominated dcpo's form a strictly larger class than that of dominated ones. As we will see, our example is given in Example~\ref{notdominated}, which relies on two steps of preliminary constructions that are based on Isbell's complete lattice and Johnstone's dcpo.
\begin{example}\label{I}
	Let $\mathbb{N}^{\mathbb{N}}$ denote the maps from $\mathbb{N}$ to $\mathbb{N}$, $S = \mathbb{N} \times \mathbb{N}$, $T = \mathbb{N}^{\mathbb{N}}\times \mathbb{N}$, $L = S \cup T\cup \{\top \}$, where $\mathbb{N}$ is the set of positive natural numbers. An order $\leq $ on $L$ is defined as follows:
	
	$(m_{1},n_{1})\leq (m_{2},n_{2})$ if and only if:
	
	\begin{itemize}
	\item $(m_{2},n_{2}) =\top $, $(m_{1},n_{1}) \in  L$; 
	\item $m_{1}=m_{2}$, $n_{1}\leq n_{2}$.
	\end{itemize}
	Next, let $P = L\times \mathbb{R}$, where $\mathbb{R}$ is the set of real numbers. Then there exists an injection $i : \mathbb{R}\times \mathbb{R}\times \mathbb{N}^{\mathbb{N}}\times \mathbb{N} \rightarrow \mathbb{R}$ such that $i(s,t, f, k) >s,t$ for any $(s,t, f, k) \in \mathbb{R}\times \mathbb{R}\times \mathbb{N}^{\mathbb{N}}\times \mathbb{N}$ with $ t>s$. We write $((n, j), k)$ by $(n, j, k)$ for any $((n, j), k) \in  P$ and $\{(n,i, r) : (n,i) \in  L\}$ by $L_{r}$ for any $r \in  \mathbb{R}$. An order $\leq $ on $P$ is defined as:   $(n_{1},i_{1}, j_{1}) \leq (n_{2},i_{2}, j_{2})$ if and only if:
	\begin{itemize}
	\item $j_{1}=j_{2}$, $(n_{1},i_{1})\leq (n_{2},i_{2})$ in $L_{j_{1}}$;
	\item $(n_{2},i_{2})=\top$. If $(n_{1},i_{1})\in S$, then there exists $t_{1}\in \mathbb{R}$, $f\in \mathbb{N}^{\mathbb{N}},k\in \mathbb{N}$ such that $j_{2}=i(j_{1},t_{1},f,k)$, $t_{1}>j_{1}$ and $(n_{1},i_{1})=(k,f(k))$. Else, $(n_{1},i_{1})\in T$, then there exists $t_{2}\in \mathbb{R}$ such that $j_{2}=i(t_{2},j_{1},n_{1},i_{1})$ and $j_{1}>t_{2}$.
       \end{itemize}
\end{example}

\begin{lemma}
The poset $P$ is a dcpo with the fact that there are only finitely many minimal upper bounds for each pair of elements in $P$.
\end{lemma}

	\begin{proof}
		The proof is similar to \cite[Fact 4.1]{Miao1}.
	\end{proof}

\begin{example}\label{J}
	Let $\mathbb{J}_{\infty}=\mathbb{N}\times (\mathbb{N}\times (\mathbb{N}\cup\{\omega\}))$, where $\mathbb{N}$ is the set of all positive natural numbers. We write $(k, (n, j))$ by $(k, n, j)$ for any $(k,(n, j)) \in  \mathbb{J}_{\infty}$ and $\{(k,n,i) : (n,i) \in  \mathbb{N}\times (\mathbb{N}\cup\{\omega\})\}$ by $\mathbb{J}_{k}$ for any $k\in \mathbb{N}$. An order $\leq $ on $\mathbb{J}_{\infty}$ is defined as: $(k,m,n)\leq (c,a,b)$ if and only if:
	\begin{itemize}
	\item $k=c$, $m=a$, $n\leq b$;
	\item $k\leq c$, $n\leq a$, $b=\omega$.
	\end{itemize}
	Then $C(\Gamma \mathbb{J}_{\infty})=\{\da x: x\in \mathbb{J}_{\infty}\}\cup \{\bigcup_{i=1}^{n}\mathbb{J}_{i}: n\in \mathbb{N}\}\cup \{\mathbb{J}_{\infty}\}$ (We depict $\mathbb{J}_{\infty}$  as in Figure $3$). 
	\begin{proof}
		
$\mathbf{Claim~1}$: rhs $\subseteq$ lhs.

By Corollary \ref{principle}, we know that $\{\da x: x\in \mathbb{J}_{\infty}\}$ is contained in $C(\Gamma \mathbb{J}_{\infty})$. For any $n\in \mathbb{N}$, $\bigcup_{i=1}^{n}\mathbb{J}_{i}$ is Scott closed. It remains to prove that $\bigcup_{i=1}^{n}\mathbb{J}_{i}$ and $\mathbb{J}_{\infty}$ are $C$-compact elements of $\Gamma \mathbb{J}_{\infty}$. To this end, let $\mathcal{A}\in \Gamma(\Gamma \mathbb{J}_{\infty})$ with $\bigcup_{i=1}^{n}\mathbb{J}_{i}\subseteq \bigcup\mathcal{A}$. For any $(n,m,\omega)\in \mathbb{J}_{n}$, the result $\da (n,m,\omega)\in \mathcal{A}$ follows obviously. Observing the order on $\mathbb{J}_{\infty}$, we get that the Scott closed set $A_{m}=\{(c,a,b)\in\bigcup_{i=1}^{n}\mathbb{J}_{i}: c\leq n, a\in \mathbb{N}, b\leq m\}\subseteq \da (n,m,\omega)\in \mathcal{A}$. Then $A_{m}$ belongs to $\mathcal{A}$ for all $m\in \mathbb{N}$.
		Note that $(A_{m})_{m\in \mathbb{N}}$ is a directed subset of $\mathcal{A}$. Then $\sup_{m\in \mathbb{N}}A_{m}=\bigcup_{i=1}^{n}\mathbb{J}_{i}\in \mathcal{A}$ as $\mathcal{A}$ is Scott closed. Therefore, $\bigcup_{i=1}^{n}\mathbb{J}_{i}$ is a $C$-compact element of $\Gamma \mathbb{J}_{\infty}$. The proof of that $\mathbb{J}_{\infty}\in C(\Gamma \mathbb{J}_{\infty})$ runs similarly.
		
		$\mathbf{Cliam~2}$: lhs $\subseteq$ rhs.
		
	Let $A\in C(\Gamma \mathbb{J}_{\infty})$. For the sake of a contradiction we assume that $A\notin \{\da x: x\in \mathbb{J}_{\infty}\}\cup \{\bigcup_{i=1}^{n}\mathbb{J}_{i}: n\in \mathbb{N}\}\cup \{\mathbb{J}_{\infty}\}$. 
	\begin{itemize}
	\item[]  $\mathbf{Claim~2.1}$: $\max A\subseteq \max \mathbb{J}_{\infty}$.
		
		Assume  that there is $a\in \max A\backslash \max\mathbb{J}_{\infty}$.  Then $A=\da \max A\subseteq \da a\cup \da (\max A\backslash \{a\})$. One sees immediately that $\da \max A\backslash \{a\}$ is Scott closed from the order of $\mathbb{J}_{\infty}$. Because $C(\Gamma \mathbb{J}_{\infty})$ is contained in $IRR(\mathbb{J}_{\infty})$, $A$ is irreducible. It follows that $A=\da a$ or $A=\da (\max A\backslash \{a\})$. The assumption that $A\notin\{\da x: x\in \mathbb{J}_{\infty}\}$ implies  that $A=\da (\max A\backslash \{a\})$. This means that $a\in \max A\subseteq \da (\max A\backslash \{a\})$, which yields that $a\in \max A\backslash\{a\}$. A contradiction.
	\item[]  $\mathbf{Claim~2.2}$: Define $A_{\mathbb{N}}=\{n\in \mathbb{N}: A\cap \mathbb{J}_{n}\neq \emptyset\}$. Then $A_{\mathbb{N}}$ is infinite. 
		
		Suppose $A_{\mathbb{N}}$ is finite. Let $n_{0}=\max A_{\mathbb{N}}$. According to $A\notin \{\bigcup_{i=1}^{n}\mathbb{J}_{i}: n\in \mathbb{N}\}$, we have that $\max A\cap \mathbb{J}_{n_{0}}$ is finite. Then $A\subseteq \da(\max A\cap\mathbb{J}_{_{n_{0}}} )\cup \bigcup_{i=1}^{n_{0}-1}\mathbb{J}_{i}$. The set $\da(\max A\cap\mathbb{J}_{n_{0}} )$ is Scott closed as  $\max A\cap \mathbb{J}_{n_{0}}$ is finite. It is easy to see that $\bigcup_{i=1}^{n_{0}-1}\mathbb{J}_{i}$ is Scott closed. From the irreducibility of $A$ and the fact that $n_{0}$ belongs to $A_{\mathbb{N}}$, we know that $A=\da(\max A\cap \mathbb{J}_{n_{0}})=\bigcup_{a\in \max A\cap \mathbb{J}_{n_{0}} } \da a$. This indicates that there exists $a\in \max A\cap  \mathbb{J}_{n_{0}}$ such that $A=\da a$, which is a contradiction. Hence, $A_{\mathbb{N}}$ is infinite.

       \end{itemize}
		
		The assumption that $A\neq \mathbb{J}_{\infty}$ implies that there exists a minimum natural number $k_{0}\in \mathbb{N}$ such that $\max A\cap \mathbb{J}_{k}$ is finite for any $k\geq k_{0}$. If not, for any $n\in \mathbb{N}$, there exists $k_n\geq n$ such that $\max A\cap \mathbb{J}_{k_n}$ is infinite. As $A$ is Scott closed, we have that $\mathbb{J}_n\subseteq \bigcup_{i\leq k_n}\mathbb{J}_i\subseteq A$. This yields $\mathbb{J}_{\infty}=A$, a contradiction. We write $B=\{m\in \mathbb{N}:\exists (n,m,\omega)\in A,n\geq k_{0}\}$, and distinguish the following two cases for $B$.
		
		Case $1$: $B$ is finite. Then $A\subseteq \da (\max A\cap \mathbb{J}_{k_{0}})\cup \da(\max A\backslash \mathbb{J}_{k_{0}})$. Since $\max A\cap \mathbb{J}_{k}$ is finite for any $k\geq k_{0}$, we have $\da(\max A\cap \mathbb{J}_{k_{0}})$ is Scott closed. The finiteness of $B$ ensures that $\da(\max A\backslash \mathbb{J}_{k_{0}})$ is Scott closed. Similar to the proof of Claim 2.2, we know that $A=\da x$ for some $x\in \mathbb{J}_{\infty}$. A contradiction.
		
		Case $2$: $B$ is infinite. Note that $\max A\cap \mathbb{J}_i$ is infinite for any $i\leq k_0-1$. Then $\bigcup_{i=1}^{k_{0}-1}\mathbb{J}_{i}\subseteq A$ as $A$ is Scott closed. It follows that for any $(c, a,b)\in \mathbb{J}_{\infty}\backslash \max \mathbb{J}_{\infty}$ with $c\geq k_{0}$, the set $C=\{m\in B: m\geq b\}$ is infinite from the infiniteness of $B$. Define $E=\{k\geq k_0: \exists (k,m,\omega)\in A, m\in C\}$. If E is finite, then there is $k\in E$ such that $\max A\cap \mathbb{J}_{k}$ is infinite. It follows that $k_{0}\geq k+1$ as $k_0$ is the minimum natural number with the property that $\max A\cap \mathbb{J}_{k}$ is finite for any $k\geq k_{0}$. This contradicts that $k\in E$. Else, $E$ is infinite. Then we could find $k\in E$ with $k\geq c$. This guarantees the existence of $(k,m,\omega)\in A$ such that $(c,a,b)\leq (k,m,\omega)$, which yields that $(c,a,b)\in A$ for any $(c,a,b)\in \mathbb{J}_{\infty}\backslash \max\mathbb{J}_{\infty}$ with $c\geq k_0$. Then we would know that $A=\mathbb{J}_{\infty}$, again from the Scott closedness of $A$. But that violates the assumption that $A\neq \mathbb{J}_{\infty}$. In summary, we have that Claim~2 holds.
		
 In a conclusion, $C(\Gamma \mathbb{J}_{\infty})=\{\da x: x\in \mathbb{J}_{\infty}\}\cup \{\bigcup_{i=1}^{n}\mathbb{J}_{i}: n\in \mathbb{N}\}\cup \{\mathbb{J}_{\infty}\}$.
	\end{proof}
\end{example}

\begin{figure}[H]
	\centering
	\begin{tikzpicture}[line width=0.5pt,scale=1.1]
		\fill[black] (0,0) circle (2pt);
		\fill[black] (0,1) circle (2pt);
		\fill[black] (0,2) circle (2pt);
		\fill[black] (0,4) circle (2pt);
		\draw (0,0)--(0,2);
		\draw [dashed](0,2)--(0,4);
		\fill[black] (1,0) circle (2pt);
		\fill[black] (1,1) circle (2pt);
		\fill[black] (1,2) circle (2pt);
		\fill[black] (1,4) circle (2pt);
		\draw (1,0)--(1,2);
		\draw [dashed](1,2)--(1,4);
		\fill[black] (2,0) circle (2pt);
		\fill[black] (2,1) circle (2pt);
		\fill[black] (2,2) circle (2pt);
		\fill[black] (2,4) circle (2pt);
		\draw (2,0)--(2,2);
		\draw [dashed](2,2)--(2,4);
		\draw [dashed](-1,4)--(0,4);
		\draw (0,0.4)--(2,0.4);
		\draw (0,-0.4)--(2,-0.4);
		\draw (2,-0.4) arc (-90:90:0.4);
		\draw [dashed](-1,0.4)--(0,0.4);
		\draw [dashed](-1,-0.4)--(0,-0.4);
		
		\draw (0,0.6)--(2,0.6);
		\draw (0,1.4)--(2,1.4);
		\draw (2,0.6) arc (-90:90:0.4);
		\draw [dashed](-1,1.4)--(0,1.4);
		\draw [dashed](-1,0.6)--(0,0.6);
		
		\draw (0,1.6)--(2,1.6);
		\draw (0,2.4)--(2,2.4);
		\draw (2,1.6) arc (-90:90:0.4);
		\draw [dashed](-1,2.4)--(0,2.4);
		\draw [dashed](-1,1.6)--(0,1.6);
		
		\fill[black] (6,0) circle (2pt);
		\fill[black] (6,1) circle (2pt);
		\fill[black] (6,2) circle (2pt);
		\fill[black] (6,4) circle (2pt);
		\draw (6,0)--(6,2);
		\draw [dashed](6,2)--(6,4);
		\fill[black] (4,0) circle (2pt);
		\fill[black] (4,1) circle (2pt);
		\fill[black] (4,2) circle (2pt);
		\fill[black] (4,4) circle (2pt);
		\draw (4,0)--(4,2);
		\draw [dashed](4,2)--(4,4);
		\fill[black] (5,0) circle (2pt);
		\fill[black] (5,1) circle (2pt);
		\fill[black] (5,2) circle (2pt);
		\fill[black] (5,4) circle (2pt);
		\draw (5,0)--(5,2);
		\draw [dashed](5,2)--(5,4);
		\draw [dashed](4,4)--(3,4);
		\draw (4,0.4)--(6,0.4);
		\draw (4,-0.4)--(6,-0.4);
		\draw (6,-0.4) arc (-90:90:0.4);
		\draw [dashed](4,0.4)--(3,0.4);
		\draw [dashed](4,-0.4)--(3,-0.4);
		
		\draw (4,0.6)--(6,0.6);
		\draw (4,1.4)--(6,1.4);
		\draw (6,0.6) arc (-90:90:0.4);
		\draw [dashed](4,1.4)--(3,1.4);
		\draw [dashed](4,0.6)--(3,0.6);
		
		\draw (4,1.6)--(6,1.6);
		\draw (4,2.4)--(6,2.4);
		\draw (6,1.6) arc (-90:90:0.4);
		\draw [dashed](4,2.4)--(3,2.4);
		\draw [dashed](4,1.6)--(3,1.6);
		
		\fill[black] (8,0) circle (2pt);
		\fill[black] (8,1) circle (2pt);
		\fill[black] (8,2) circle (2pt);
		\fill[black] (8,4) circle (2pt);
		\draw (8,0)--(8,2);
		\draw [dashed](8,2)--(8,4);
		\fill[black] (9,0) circle (2pt);
		\fill[black] (9,1) circle (2pt);
		\fill[black] (9,2) circle (2pt);
		\fill[black] (9,4) circle (2pt);
		\draw (9,0)--(9,2);
		\draw [dashed](9,2)--(9,4);
		\fill[black] (10,0) circle (2pt);
		\fill[black] (10,1) circle (2pt);
		\fill[black] (10,2) circle (2pt);
		\fill[black] (10,4) circle (2pt);
		\draw (10,0)--(10,2);
		\draw [dashed](10,2)--(10,4);
		\draw [dashed](8,4)--(7,4);
		\draw (8,0.4)--(10,0.4);
		\draw (8,-0.4)--(10,-0.4);
		\draw (10,-0.4) arc (-90:90:0.4);
		\draw [dashed](8,0.4)--(7,0.4);
		\draw [dashed](8,-0.4)--(7,-0.4);
		
		\draw (8,0.6)--(10,0.6);
		\draw (8,1.4)--(10,1.4);
		\draw (10,0.6) arc (-90:90:0.4);
		\draw [dashed](8,1.4)--(7,1.4);
		\draw [dashed](8,0.6)--(7,0.6);
		
		\draw (8,1.6)--(10,1.6);
		\draw (8,2.4)--(10,2.4);
		\draw (10,1.6) arc (-90:90:0.4);
		\draw [dashed](8,2.4)--(7,2.4);
		\draw [dashed](8,1.6)--(7,1.6);
		\draw (2.4,0)--(6,4);
		\draw (2.4,1)--(5,4);
		\draw (2.4,2)--(4,4);
		\draw[red] (6.4,0)--(10,4);
		\draw [red](6.4,1)--(9,4);
		\draw [red](6.4,2)--(8,4);
		\draw [red](2.4,0)--(10,4);
		\draw [red](2.4,1)--(9,4);
		\draw [red](2.4,2)--(8,4);
		
		\draw [dashed](11,0)--(12,0);
		\draw [dashed](11,4)--(12,4);
		\draw [dashed](11,3)--(12,3);
		
		\node[font=\footnotesize] (1) at(-0.5,0.75) {$(1,3,2)$};
		\node[font=\footnotesize] (2) at(6,4.25) {$(2,1,\infty)$};
		\node[font=\footnotesize] (3) at(10,4.25) {$(3,1,\infty)$};
		\node[font=\footnotesize] (A1) at(1,-0.6) {$\mathbb{J}_1$};
		\node[font=\footnotesize] (A2) at(5,-0.6) {$\mathbb{J}_2$};
		\node[font=\footnotesize] (A3) at(9,-0.6) {$\mathbb{J}_3$};
	\end{tikzpicture}
	 \scriptsize  \\  Figure 3: The order between different blocks on $\mathbb{J}_{\infty}$. 
	\end{figure}
	Now we construct a weakly dominated dcpo $M$ based the above dcpo's $P$ and $\mathbb{J}_{\infty}$, which fails to be dominated. 
\begin{example}\label{notdominated}
	Let $M=P\cup \mathbb{J}_{\infty}$. Fix $\{a_{\mathbb{N}}: n\in \mathbb{N} \}\subseteq \mathbb{R}\backslash i(\mathbb{R}\times \mathbb{R}\times \mathbb{N}^{\mathbb{N}}\times \mathbb{N})$, where the function~$i$ is defined as in Example~\ref{I}. An order $\leq $ on $M$ is defined as $x\leq y$ if and only if:
	\begin{itemize}
	\item $x\leq y$ in $P$; 
	\item $x\leq y$ in $\mathbb{J}_{\infty}$;
	\item $y=(a_{n},\top)\in P$, $x=(c,a,b)\in \mathbb{J}_{\infty}$, $c\leq n$.
	\end{itemize}
	Then $(M,\leq)$ is weak dominated, but not dominated.

	\begin{proof}
		$\mathbf{Claim ~1}$: $M$ is not dominated.
		
		Similar to the  argument of \cite[Remark 4.1]{Miao1}, $M$ is an irreducible closed set. Set $ \nabla^* M=\{B\in IRR(M): \exists x\in M, s.t.~B\subseteq \da x\}$. Notice that $\{\bigcup_{i=1}^{n}\mathbb{J}_{i}: n\in \mathbb{N}\}$ is a directed subset of $ \nabla^* M$, but $\sup_{n\in \mathbb{N}}\bigcup_{i=1}^{n}\mathbb{J}_{i}=\mathbb{J}_{\infty}\notin \nabla^* M$. Hence, by definition, $M$ is not dominated.
		
		$\mathbf{Claim ~2}$: $C(\Gamma M)=\{\da x: x\in M\}\cup \{\bigcup_{i=1}^{n}\mathbb{J}_{i}: n\in \mathbb{N}\}\cup \{\mathbb{J}_{\infty}\}$.
		
The right side of the equation is obviously contained in the left side by Example \ref{J}. Conversely, let $A\in C(\Gamma M)$. By way of contradiction  we assume that $A\notin \{\da x: x\in M\}\cup \{\bigcup_{i=1}^{n}\mathbb{J}_{i}: n\in \mathbb{N}\}\cup \{\mathbb{J}_{\infty}\}$.
		
		We distinguish the cases whether $A$ is contained in $\mathbb{J}_{\infty}$:
		
		Case $1$: $A\subseteq \mathbb{J}_{\infty}$. In this case we claim that $A$ is a $C$-compact element of $\Gamma \mathbb{J}_{\infty}$. To this end, let $\mathcal{A}\in \Gamma(\Gamma \mathbb{J}_{\infty})$ with $A\subseteq \bigcup\mathcal{A}$. Then the fact that $\mathbb{J}_{\infty}$ is a Scott closed set of $M$ indicates that $\mathcal{A}$ is a Scott closed subset of $\Gamma M$. It follows that $A\in \mathcal{A}$ as $A\in C(\Gamma M)$. So we have $A$ is a $C$-compact element of $\Gamma \mathbb{J}_{\infty}$, that is, $A\in C(\Gamma \mathbb{J}_{\infty})=\{\da x: x\in \mathbb{J_{\infty}}\}\cup \{\bigcup_{i=1}^{n}\mathbb{J}_{i}: n\in \mathbb{N}\}\cup \{\mathbb{J}_{\infty}\}$. That is a contradiction.
		
		Case $2$: $A\not \subseteq \mathbb{J}_{\infty}$. Then $A\cap P\neq \emptyset$. We claim that $A=\da (A\cap P)$. Note that $A\subseteq \da (A\cap P)\cup \mathbb{J_{\infty}}$. By the definition of $M$, one sees immediately that $\da (A\cap P)$ is a Scott closed set of $M$. Since $A$ is a $C$-compact element of $\Gamma(M)$, we know that $A$ is irreducible. This means that $A\subseteq \da(A\cap P)$ or $A\subseteq \mathbb{J}_{\infty}$. The assumption that $A\not \subseteq \mathbb{J}_{\infty}$ implies that $A\subseteq \da(A\cap P)$. In other words, $A=\da (A\cap P)$. Similar to the proof of Example \ref{J}, we have $\max A\subseteq \max P$. Define $\mathcal{A}=\da (\{\da x: x\in \max A\})\cup \da \mathbb{J}_{\infty}$. It is obvious that $\mathcal{A}$ is a lower set of $\Gamma M$. We want to show that $\mathcal{A}$ is Scott closed. To this end, let $(A_{i})_{i\in I}$ be a non-trivial chain of $\mathcal{A}$. Now we distinguish the following two cases for $(A_{i})_{i\in I}$.
		
		Case $2.1$: $(A_{i})_{i\in I}\subseteq \da \mathbb{J_{\infty}}$. Then the fact that $\mathbb{J}_{\infty}$ is Scott closed concludes that $\sup_{i\in I}A_i=cl(\bigcup_{i\in I}A_{i})\in\da \mathbb{J_{\infty}} \subseteq \mathcal{A}$. 
		
		Case $2.2$: $(A_{i})_{i\in I}\not \subseteq \da \mathbb{J_{\infty}}$. Then there exists $i_{0}\in I$ such that $A_{i}\cap P\neq \emptyset$ for any $i\geq i_{0}$. This means that there exists $b_{i}\in \mathbb{R}$ such that $A_{i}\subseteq \da(b_{i},\top)$ for all $i\geq i_{0}$. Now we need to further distinguish the following two cases. 
		
		Case $2.2.1$: $A_{i}\subseteq P$ for any $i\geq i_{0}$. If $\{b_{i}\mid i\geq i_{0}\}$ is a finite set, then there exists $b_{i}$ such that $A_{i}\subseteq \da (b_{i},\top )$ for any $i\in I$. It follows that $\sup_{i\in I}A_{i}\in \da (\da (b_{i},\top))\subseteq \mathcal{A}$. Otherwise, $\{b_{i}\mid i\geq i_{0}\}$ is an infinite set. 
		Then $A_{i}\subseteq \bigcap_{j\geq i}\da(b_{j},\top)$ for any $i\geq i_{0}$. In light of Example \ref{I}, we have that $P$ is a dcpo with the condition that there are only finitely many minimal upper bounds of each pair of elements in $P$, which yields that for any $i\geq i_0$, $A_{i}$ must be the form of $\{(k,m,n): n\leq n_{i}\}$, for some fixed $n_{i}\in \mathbb{N}, m\in \mathbb{N}\cup \mathbb{N}^{\mathbb{N}}$ and $k\in \mathbb{R}$, and this reveals that $\{(k,m,n): n\in \mathbb{N}\}$ is a directed subset of $A$ since $A$ is a lower set. Scott closedness of $A$ then implies that $(k,m,\omega)\in A$. Hence, $\da(k,m,\omega)=cl(\bigcup_{i\in I}A_i)=\sup_{i\in I}A_{i}\in \mathcal{A}$. 
		
		Case $2.2.2$: there exists $i_{1}\geq i_{0}$ such that $A_{i_{1}}\cap \mathbb{J}_{\infty}\neq \emptyset$. Then $b_{i}\in \mathbb{R}\backslash i(\mathbb{R}\times \mathbb{R}\times \mathbb{N}^{\mathbb{N}}\times \mathbb{N})$ for any $i\geq i_{1}$. For any $i,j\geq i_{1}$ with $j\geq i$, we know $\emptyset\neq A_{i}\cap P\subseteq A_{j}\cap P$. This implies that $b_{i}=b_{j}$ by the order defined on $M$, which yields that $b_i = b_0$ for any $i\geq i_1$. It turns out that $A_i\subseteq \da(b_0,\top)$ for any $i\in I$. Thus, $\sup _{i\in I}A_{i}\subseteq \da(b_0,\top)\in \mathcal{A}$. 
		
		In the both sub-cases, hence in the Case 2.2,  $\mathcal{A}$ is Scott closed. 
		
		Now in the Case 2, we then know that $A\in \mathcal{A}$ since $A\in C(\Gamma M)$. This gives that $A=\da x$ for some $x\in A$ or $A\subseteq \mathbb{J}_{\infty}$. A contradiction. Hence we finish proving Claim 2. 
		
		Now it is trivial to check that $\nabla A$ is Scott closed for $A\in C(\Gamma M)$. So Indeed, $(M,\leq)$ is weak dominated, but not dominated.
	\end{proof}
\end{example}

\section*{Acknowledgement}
{This work is supported by the National Natural Science Foundation of China (No.12401596, No.12231007 and No.12371457)}.

\bibliographystyle{plain}

\end{document}